\documentclass[journal,transmag]{IEEEtran}
\ifCLASSINFOpdf
\else
\fi

\usepackage{graphicx}
\usepackage[english]{babel}
\usepackage{subfigure}
\usepackage{amsthm}
\usepackage{amsmath} 
\usepackage{eufrak}
\usepackage{babel} 
\usepackage{mdwmath}
\usepackage{mdwtab}
\usepackage{mdframed}
\usepackage{booktabs}
\usepackage{times}
\usepackage{mathrsfs}
\usepackage{cite}
\usepackage{comment}
\usepackage{color}
\usepackage{mdframed}
\usepackage{multirow}
\usepackage{etoolbox}
\usepackage{amsfonts}

\usepackage {algorithm}
\usepackage{algpseudocode}
\algblockdefx{Event}{EndEvent}[1]%
{\textbf{upon event }#1 \textbf{do }}%
{}
\usepackage{wrapfig}

\renewenvironment{thebibliography}[1]{%
	\begin{oldthebibliography}{#1}%
		\setlength{\itemsep}{-.02ex}%
	}%
	{%
	\end{oldthebibliography}%
}

\renewcommand{\baselinestretch}{0.96}

\usepackage{xcolor}
\definecolor{red}{rgb}{0,0,0}

\begin{document}
%

\title{\textit{RT-ByzCast}: Byzantine-Resilient Real-Time Reliable Broadcast} 


\author{\IEEEauthorblockN{David Kozhaya$^{1}$, \textcolor{red}{J\'er\'emie Decouchant$^{2}$} and
Paulo Esteves-Verissimo$^{2}$,~\IEEEmembership{Fellow,~IEEE,~ACM}}
\IEEEauthorblockA{\textcolor{red}{$^{1}$ABB Corporate Research, Baden-Dattwil, Switzerland}\\ $^{2}$Interdisciplinary Centre for Security, Reliability and Trust (SnT), University of Luxembourg, Luxembourg}
}

%



\IEEEtitleabstractindextext{%
\begin{abstract}
\textcolor{red}{Today's cyber-physical systems face various impediments to
achieving their intended goals, 
namely, communication
uncertainties and faults, relative to the increased
integration of networked and wireless devices, hinder the
synchronism needed to meet real-time deadlines. Moreover,
being critical, these systems are also exposed to significant
security threats. This threat combination increases the risk of physical
damage.}
This paper addresses these problems by studying how to build the first real-time Byzantine reliable broadcast protocol (RTBRB) tolerating network uncertainties, faults, and attacks. Previous literature describes either real-time reliable broadcast protocols, or asynchronous (non real-time) Byzantine~ones.

We first prove that it is impossible to implement RTBRB using traditional distributed computing paradigms, e.g., where the error/failure detection mechanisms of processes are decoupled from the broadcast algorithm itself, even with the help of the most powerful failure detectors. We circumvent this impossibility by proposing \textit{RT-ByzCast}, an algorithm based on aggregating digital signatures in a sliding time-window and on empowering processes with self-crashing capabilities to mask and bound losses. We show that \textit{RT-ByzCast} (i) operates in real-time by proving that messages broadcast by correct processes are delivered within a known bounded delay, and (ii) is reliable by demonstrating that correct processes using our algorithm crash themselves with a negligible probability, even with message loss rates as high as~$60\%$. 
\end{abstract}

}

\maketitle

\IEEEdisplaynontitleabstractindextext

%
\IEEEpeerreviewmaketitle

	\newtheorem{theorem}{Theorem}
\newtheorem{corollary}{Corollary}
\newtheorem{assumption}{Assumption}
\newtheorem{lemma}{Lemma}
\newtheorem{definition}{Definition}
\newtheorem{conjecture}{Conjecture}
\newtheorem{observation}{Observation}
\newtheorem{remark}{Remark}
\newtheorem{conclusion}{Conclusion}
\newcommand{\todo}[1]{\color{green} TODO: {#1}\color{black}}
\definecolor{orange}{rgb}{1,0.5,0}

\section{Introduction}\label{intro}
Many of today's physical structures are governed by automated operations that are typically conducted and controlled \textcolor{red}{by multiple sensing, computing, and communication devices. Generally, such physical structures including their distributed control are known as \textit{cyber-physical systems} (CPS)~\cite{controlarea}.} 
In this paper, we address two problems relevant to modern CPS, found for example in large continuous process plants, manufacturing shop-floors, or power grid~installations.

The increasing integration of sensors and actuators, networked and often wireless~\cite{endtoendreal-time,controlloss1,eors,D1.1}, \textcolor{red}{as well as their scale, introduce uncertainties and faults in communication, which hamper the necessary synchronism to meet real-time deadlines, fundamental for various CPS applications~\cite{controlloss1,eors,D1.1,netfail}.} 

Worse, most of those processes are in essence critical infrastructures. Hence, they are likely to be \textcolor{red}{targeted by motivated attackers~\cite{sensorsecurity,sensorsecurity1}, to impact the timeliness of communications --- for example by maliciously worsening the impact of the above mentioned uncertainties and faults --- and the sheer correctness of messages.} These two facets create a completely new scenario, \textcolor{red}{where the traditional approach to building real-time communication no longer works.} The consequences of system failure can range from loss of availability to physical~damage.

\textcolor{red}{In this paper, we simultaneously address the three threats --- uncertainty, faults, and attacks --- introducing \textit{RT-ByzCast}, to the best of our knowledge the first Byzantine-resilient real-time reliable broadcast protocol.} Previous literature on fault-tolerance contains either real-time reliable broadcast protocols~\cite{jiter,yair1,endtoendreal-time,real-timemonitor,real-timecontrol}, or asynchronous (non real-time) Byzantine ones~\cite{Bracha:1985,Bracha:1987,raynal2015,Schneider:1984:BGA:190.357399,Dolev:1981:BGS:891722}. On the one hand, real-time reliable broadcast protocols cannot tolerate attacks: such algorithms would fail, even if a single process is compromised. \textcolor{red}{Using naive platform restarts or cryptography alone is insufficient}, and hence does not solve the problem~\cite{Guerraoui:2006:IRD:1137759,restart}. On the other hand, existing asynchronous Byzantine solutions do not support real-time. \textcolor{red}{Yet, critical CPS applications must have information disseminated correctly, reliably and in real-time, e.g., as in power system control applications that need} to open and close multiple distributed circuit breakers based on thresholds and sensor data~\cite{Dav}.


We approach the aforementioned threats first by formally specifying the needed broadcast properties via an abstraction, which we call \textbf{R}eal-\textbf{T}ime \textbf{B}yzantine \textbf{R}eliable \textbf{B}roadcast (RTBRB). We then prove that in systems where processes can be Byzantine and communication can simultaneously fail, it is impossible to implement the RTBRB abstraction using traditional distributed computing assumptions and techniques, e.g., where \textcolor{red}{the 
error/failure detection mechanisms of processes are oblivious of the broadcast algorithm. We also prove that this holds even with the help of sophisticated failure detectors that perfectly reveal the identities of all faulty and untimely~processes.}

We propose an algorithm, which we call \textit{RT-ByzCast}, \textcolor{red}{which} can cope with intermediate timing violations and uncertainties while still providing a reliable end result in real-time. We specifically prove that our algorithm indeed implements all properties of the RTBRB abstraction. The main idea underlying our \textit{RT-ByzCast} algorithm is the use of a temporal and spatial diffusion mechanism of signed messages over a sliding time-window. This diffusion mechanism is augmented with a scheme to aggregate signatures corresponding to those messages. Hence, the time-window constitutes a slack. During that slack, the aggregation and diffusion mechanism tries to mask unanticipated violations, e.g., those caused by communication losses, by having processes repeatedly aggregate signatures relative to the same message and disseminate that message again (temporal diffusion) and to other processes (spatial~diffusion).


In order to ensure safety, when losses cannot be masked in time using our aggregation and diffusion scheme, our algorithm triggers processes suffering such losses to crash themselves. Doing so, \textit{RT-ByzCast} bounds the amount of time a correct process can lose communication with a quorum of the network. \textcolor{red}{This self-crashing capability allows our algorithm to circumvent our impossibility result affecting traditional distributed assumptions and techniques: we use process self-crashing to make sure that the view of correct processes is consistent between the error/failure detection mechanisms and our broadcast algorithm. \textit{RT-ByzCast} can affect which processes are correct and which are not can revive processes that have self-crashed themselves}.
Our \textit{RT-ByzCast} algorithm relies on quorums. The reader might wonder that our principle, albeit safe, may lead the system to a complete shutdown, by having ``too many" processes 
crash themselves. To this end, we perform a thorough simulation of the probability of a process crashing itself in our \textit{RT-ByzCast} algorithm. We demonstrate that even for high loss rates (60$\%$ loss rate), the probability that a process crashes itself in \textit{RT-ByzCast} is negligible for systems with more than 5 processes. We show as well that as the system size increases, the probability of a process crashing itself becomes asymptotic to 0 for any loss rate; hence the system liveness is almost always guaranteed. 
This matches the typical applications of the target CPS systems we envisage, e.g., the smart grid, where computing sensing and communication devices are deployed in large scale. Moreover, we devise a method that allows systems using our algorithm to tolerate any number of correct processes crashing themselves without shutting down. \textcolor{red}{Our method is based on a scheme for detecting correct processes that crash themselves and over provisioning the number of processes (replicas)~used.} 

\textcolor{red}{Another concern of the reader might arise regarding the expected performance of our algorithm. We show that \textit{RT-ByzCast} can meet the timing constraints of a large class of typical CPS applications, mainly in SCADA and IoT areas: power system automation and substation automation applications (time constants $\leq 100$~ms); slow speed auto-control functions ($\leq 500$~ms); continuous control applications ($\leq 1$~s); or (iii) operator commands of SCADA applications~($\leq 2$ s).}




In summary, the main contributions of this paper are:
\begin{enumerate}
	\item Theoretical proofs showing that when processes and links can fail, it is impossible to implement RTBRB under traditional distributed computing paradigms that consider processes' correctness checks as oblivious of the broadcast algorithm. We prove that the impossibility holds even when the identity of all faulty and unresponsive processes is known.
	\item \textit{RT-ByzCast}, an algorithm that implements the RTBRB abstraction by aggregating and diffusing digital signatures in a sliding time-window to mask individual communication losses. When messages cannot be masked in time, \textit{RT-ByzCast} forces processes experiencing those losses to crash themselves.
	\item \textcolor{red}{A thorough simulation of the reliability and availability  of our \textit{RT-ByzCast} algorithm under many communication loss rates and system sizes.}
	\item \textcolor{red}{A performance evaluation showing \textit{RT-ByzCast}'s latency and network overhead.}
	\item Ways to adapt \textit{RT-ByzCast} \textcolor{red}{to handle churn.} 
\end{enumerate}

The rest of the paper is organized as follows. Section~\ref{Sysmodel} details the system model. Section~\ref{RTBRB} defines the real-time Byzantine reliable broadcast (RTBRB) abstraction. Section~\ref{feasibility of implementing RTBRB} studies the feasibility of implementing RTBRB under traditional distributed computing assumptions. Section~\ref{implementing RTBRB} presents the intuitions and details of our \textit{RT-ByzCast} algorithm. Section~\ref{Evaluation} demonstrates the reliability nad performance of our \textit{RT-ByzCast} algorithm by (i) evaluating the probability of a process crashing itself, (ii) computing the probability that the system shuts down, (iii) showing how to allow a system to tolerate the crash of correct processes without causing system shutdown, \textcolor{red}{and (iv) showing the broadcast delivery latency under different system sizes and message loss rates. Section~\ref{Dead-State Revival} discusses how to revive self-crashed processes.} Section~\ref{dynamic system} illustrates and details how our \textit{RT-ByzCast} algorithm can be adapted to handle dynamic systems. Section~\ref{RT-ByzCast applicability} discusses \textit{RT-ByzCast} application domains and systems. Finally Section~\ref{related work} and Section~\ref{conclusion} discuss existing related work and conclude the paper respectively. \textcolor{red}{For presentation purposes, we defer proofs and additional evaluations to a companion dedicated~appendix.}

\section{System and Threat Model}\label{Sysmodel}

\subsection{System Model}\label{system model}
\paragraph{Processes}
We consider a distributed system 
consisting of a set of $n>~1$ processes, denoted by $\mathit{\Pi}=\{p_1, p_2, ..., p_n\}$. Processes are synchronous, i.e., the delay for performing a local step has a fixed known bound, assumed to be negligible compared to communication delays. 

Processes have access to local clocks. For presentation simplicity, we assume that these clocks are synchronized, with a bounded skew. Using these clocks, processes define \emph{synchronous rounds} of the same fixed duration (time-triggered). \textcolor{red}{Rounds are synchronized among all processes, i.e., the start and end of a round occur at all hosts at the same time (with a bounded~skew). For a particular round, clocks can trigger an initialization signal at all processes, upon which processes are awakened and initialized as part of the system.} 

\paragraph{Communication} 
Every pair of processes is connected by two logical uni-directional links. Precisely, processes $p_i$ and $p_j$ are connected by links $l_{ij}$ and $l_{ji}$. Links can abstract a physical bus or a dedicated network \textcolor{red}{link/path}. 

\textcolor{red}{We assume that links are reliable and within the maximum delay with high probability.} This means that in any transmission attempt, where a message is sent over a link, there is a high probability that the message reaches its destination within a maximum delay $d$ after being transmitted. However, there is a small probability that reliability and timeliness are violated. Such violations exist in networks, as arguably all communication is prone to random disturbances, e.g., bad channel quality, interference, collisions, and \textcolor{red}{buffer overflows}~\cite{eors}. 

We consider both message losses and delays as omissions. That is, late messages (violating the $d$ delay assumption) are simply dropped (ignored). \textcolor{red}{This way we treat \textit{timing faults}~\cite{Carlos-almeida1998using-light-weight-31} as omissions. We define the duration of any synchronous round to be $d$}, the upper bound on a reliable and timely transmission. As such, any message sent at the beginning of round $r$, if not omitted, is assumed to be received by the end of round $r$. We consider that processes always send their messages at the beginning of a round. \textcolor{red}{In our model, it is sufficient that delay $d$ is only known by the local clocks that define rounds and not processes. Processes in some round $r$ consider late any received message with a round number $<r$.}




\subsection{Threat Model} \label{threat model}
\paragraph{Clocks}
\textcolor{red}{We assume that the local synchronized clocks of non-Byzantine nodes are secure and hence cannot be attacked. Previous work showed that such secure and synchronized clocks can be built in similar environments} using trusted components~\cite{tcb} or GPS~\cite{clock}. 

\paragraph{Processes}
We assume that some processes can exhibit arbitrary, a.k.a. \textit{Byzantine}, behavior. Byzantine nodes can abstract processes that have been compromised by attackers, or are executing the algorithm incorrectly, e.g.,  as a result of some fault (software or hardware). A Byzantine process can behave arbitrarily, e.g., it may crash, fail to send or receive messages, delay messages, send arbitrary~messages,~etc. 

We recall that in every transmission attempt a link may (with some probability) violate reliability and timeliness by dropping the message or delivering it within a delay $>d$. In both cases (dropped and delayed), that message is omitted; hence a sender needs to re-transmit that message again and face yet another risk of transmission failure. Due to omissions (losses and delays in consecutive  transmission attempts) and the required follow-up re-transmissions, the time it takes to send a message reliably from one process to another (measured from the time of the first transmission attempt) may be unbounded. So, despite links being reliable and timely with high probability, our communication system is no longer synchronous. Our system (not being synchronous) can tolerate at maximum $f=\lfloor\frac{n-1}{3}\rfloor$ Byzantine processes. Such an $f$ is proved to be the maximum number of Byzantine processes that an asynchronous system with $n$ processes can tolerate to implement any form of agreement~\cite{Dolev:1981:BGS:891722,TC1}. \textcolor{red}{$f$ is known to the processes.}


\textcolor{red}{A process that exhibits a Byzantine behavior is termed \textit{faulty}. Otherwise the process does not deviate from the specification of the algorithm and is said to be \textit{non-Byzantine}. If the algorithm specifies some non-Byzantine processes to crash themselves within the algorithm's execution, then such crashed processes are termed faulty as well (excluded by the algorithm). The rest of the processes (the non-faulty ones) are termed as \textit{correct}. More formally, consider an algorithm $\mathcal{A}$ that begins execution at some global time $t$ for a period~$\Delta T$.
\begin{definition}
	All non-Byzantine processes that do not crash themselves, in $[t,t+\Delta T]$ are correct wr.t.  $\mathcal{A}$. All other processes (Byzantine and non-Byzantine that crashed themselves) are faulty.
\end{definition}}
 

\paragraph{Communication}
Links are assumed to be faithful and authenticated, i.e., a link does not alter the content of messages. Hence, Byzantine processes cannot modify messages sent on a link connecting correct processes. Recall however, that links can, with a small probability, violate reliability and timeliness. A message transmitted on link $l_{ij}$, $\forall i\neq j$, at any time $t$ has probability $0<\epsilon_{1}<P^{o}_{ij}(t)<\epsilon_{2}<<1$ of getting lost or delayed. 

In practice, losses and delays can be correlated and may usually come in bursts, i.e., periods where messages are consecutively lost/delayed. \textcolor{red}{Burst lengths can be well anticipated (with high coverage)\footnote{\textcolor{red}{For example in control systems it is very unlikely that a network link has more than two consecutive message losses. We refer readers to studies like~\cite{gm5,yair1} that elaborate on how to estimate link burstiness.}}.} Accordingly, we denote by $k$ the maximum burst lengths (in rounds) that are expected to be observed in our network (a.k.a. omission degree). Note that this does not mean the bursts of length $>k$ rounds will not occur, it rather indicates that they seldom do and the system should still account for them. Namely, a link $l_{ij}$, $\forall i\neq j$, at any time $t$ has probability $P^{k}_{ij}(t)$ of experiencing bursts of length~$k$. $P^{k}_{ij}(t)$ is essentially a function of $P^{o}_{ij}(t)$, i.e., $P^{k}_{ij}(t)=\mathcal{F}(P^{o}_{ij}(t))$.  

\section{Real-Time Byzantine Reliable Broadcast}\label{RTBRB}

We now define the properties of the  real-time Byzantine reliable broadcast (RTBRB). Roughly speaking, RTBRB is a communication primitive that allows information to be disseminated reliably and timely in a one-to-all manner.  

For simplicity, we provide the definition of RTBRB as an abstraction that allows at most one message to be delivered/broadcast. 
This definition can be easily extended to allow processes to broadcast multiple messages, e.g., see~\cite{Guerraoui:2006:IRD:1137759,raynal2015} that resolve this using message identifiers. 
\begin{enumerate}
	\item \textbf{\textit{RTBRB-Validity:}} If a correct process broadcasts $m$, then some correct process eventually delivers $m$.
	\item \textbf{\textit{RTBRB-No duplication:}} \textcolor{red}{Every correct process delivers at most once a message.}
	\item \textbf{\textit{RTBRB-Integrity:}} If some correct process delivers a message $m$ with sender $p_i$ and process $p_i$ is correct, then $m$ was previously broadcast by $p_i$.
	\item \textbf{\textit{RTBRB-Agreement:}} If some message $m$ is delivered by any correct process, then every correct process eventually delivers $m$.
	\item \textbf{\textit{RTBRB-Timeliness:}} There exists a known $\Delta$ such that if a correct process broadcasts $m$ at real-time $t$, no correct process delivers $m$ after real~time~$t+~\Delta$.
\end{enumerate}

RTBRB provides processes with two operations, namely RTBRB-broadcast() and RTBRB-deliver(). A process broadcasts a message by invoking RTBRB-broadcast(). Similarly, a process delivers a message by invoking RTBRB-deliver().

We next compare RTBRB properties to those of existing broadcast abstractions highlighting the differences. 

\vspace*{-10pt}\subsection*{RTBRB Versus Existing Broadcast Primitives}\label{comp with exsiting broadcast}
\paragraph{Byzantine Reliable Broadcast}\label{ByzRB} 
An asynchronous  Byzantine reliable broadcast~\cite{Bracha:1985,Bracha:1987} guarantees only a subset of the properties of RTBRB (all except RTBRB-Timeliness). In this sense, as opposed to asynchronous Byzantine reliable broadcast which delivers messages only eventually (i.e., with no predictability measures) our RTBRB primitive provides a timeliness bound: a message sent by correct processes is delivered within a known fixed duration after being~broadcast.
\paragraph{Real-Time Reliable Broadcast}\label{Timely RB}
Synchronous (or real-time) reliable broadcasts~\cite{AMP,yair1} provide known fixed bounds on delivering broadcast messages. However they only handle \textit{fail-silent} process failures and hence cannot tolerate maliciousness.
\paragraph{Atomic Broadcast}\label{atomicast}
Unlike liveness in the context of (asynchronous) reliable broadcast, the RTBRB-Timeliness property (a safety property) introduces a scent of (physical) ordering, since it stipulates, for each execution, a termination event to occur ``at or before'' some $\Delta$ on the time-line. This said, one may wonder to what extent does this go into reaching a linear ordering, and thus in the direction of atomic broadcast. In RTBRB, the interleaving of broadcasts from multiple senders, namely when multiple broadcasts are issued in a period shorter than $\Delta-d$, might result in messages \textcolor{red}{delivered to different processes in a different order. RTBRB-Timeliness ensures that a message $m$ is delivered at some time in $[d,\Delta]$ after the broadcast}.

\section{Feasibility of Implementing RTBRB}\label{feasibility of implementing RTBRB}

We now study the feasibility of implementing RTBRB using traditional distributed assumptions and techniques. 

\vspace*{-10pt}\subsection{RTBRB Using Failure Detectors}
Traditionally, algorithms are designed to guarantee properties without affecting the correctness of processes. However, algorithms may rely on components, e.g., failure detectors (FDs)~\cite{Chandra:unreliableFD}, which are external software blocks that provide hints about which processes are faulty. Failure detectors encapsulate synchrony assumptions. \textcolor{red}{Thus, when added to} systems with timing uncertainties (i.e., asynchronous), failure detectors make it possible to devise modular algorithms that solve, despite network uncertainties, difficult distributed computing problems, e.g., consensus~\cite{Lamport:1998}.

In this section, we prove impossibility results showing that RTBRB cannot be implemented using traditional modular schemes relying on external software blocks, even when these blocks are the ``most powerful'' ones, such as perfect crash failure detectors, e.g.,~\cite{Chandra:unreliableFD}. \textcolor{red}{These impossibilities indicate that monitoring untimely processes in systems supported by real-time broadcast networks with timing uncertainties (timing and omission faults), should be done within the algorithm implementing RTBRB rather than externally. In short this means going against the modular failure detectors trend initiated
with~\cite{Chandra:unreliableFD} and being more in line with the model of~\cite{AMP}, hence featuring failure detectors and membership integrated with the broadcast protocol.}

We give the intuition of why this is surprisingly so, and present proofs of the
results. In this sense, we help establish, in this section, a clear
understanding of the design constraints that should be
followed/avoided when devising algorithms for distributed real-time
systems exhibiting timing uncertainties, before delving into the system design.
We first define a few essential~terms.

\begin{definition}
A component is \textbf{oblivious of an algorithm}  $\mathcal{A}$ if that component does not adapt its behavior to the decisions and actions of $\mathcal{A}$. Such a component does not take any input from $\mathcal{A}$ and hence any changes to $\mathcal{A}$ result in no impact on the behavior of that component. Moreover, messages relative to that component are assumed to be independent of $\mathcal{A}$'s messages (i.e., messages are sent separately over the network).
\end{definition}

\begin{definition}
A Perfect Crash Failure Detector ($\mathcal{P}$)~\cite{Chandra:unreliableFD} guarantees:
\begin{enumerate}
	\item \textbf{Strong completeness:} Eventually every faulty process is permanently suspected by all correct processes.
	\item  \textbf{Strong accuracy:} A correct process is never suspected. 
\end{enumerate} 
\end{definition}
\begin{assumption}\label{fdob}
	Detector $\mathcal{P}$ is oblivious of algorithm~$\mathcal{A}$ that implements RTBRB.
\end{assumption}
\begin{assumption}\label{definition A uses FD}
Algorithm $\mathcal{A}$ uses $\mathcal{P}$. 
\end{assumption}

\begin{assumption}\label{assum. no additional comp} Algorithm $\mathcal{A}$ does not use any additional hardware or software components besides $\mathcal{P}$ and the network.
\end{assumption}

\begin{theorem}\label{theorem imp using FDs}
No algorithm $\mathcal{A}$ can implement the RTBRB abstraction under Assumptions~\ref{fdob},~\ref{definition A uses FD} and~\ref{assum. no additional comp}.
\end{theorem}
\vspace*{-5pt}
The proof relies on the fact that even when knowing the identity of all correct processes, some correct processes might not be reachable in real-time (any predefined duration) through the network (due to probabilistic message losses). Traditional failure detectors~\cite{Chandra:unreliableFD} seek eventual reachability, while we seek real-time reachability. In other words, such an impossibility does not exist in time-free systems (a.k.a.  asynchronous) but merely in those systems requiring a timeliness property. The reason follows from the fact that in time-free asynchronous systems, all correct processes are eventually reachable, even with our notion of links. Reachability changes from being a liveness property (eventual reachability) in time-free systems to being a safety property (timed reachability) in a real-time~context. \textcolor{red}{The detailed proof of Theorem~\ref{theorem imp using FDs} is subsumed by the proof of the more general result of Theorem~\ref{theorem imp oblivious monitoring}.}

We investigate next the feasibility of implementing RTBRB when failure detectors can detect, not only crashed processes, but also those processes that are not reachable in time. Perhaps surprisingly, \textcolor{red}{we prove below} that the impossibility still persists, if the failure detector remains oblivious of the algorithm implementing RTBRB.

\vspace*{-10pt}\subsection{RTBRB Using Proactive Reachability Failure Detectors} \label{RTBRB with montioring} 
In order to determine which processes are alive and reachable in real-time distributed systems, a widely common practice is to use failure detectors that rely on heartbeats with a concrete timeliness specification~\cite{gm5,gm1,viewsnoop,FTFT,gm6}. 

\begin{definition}\label{reachable}
	Consider some fixed interval of $l$ rounds denoted by $\Delta_l$.
	Process $p_i$ is said to be \textbf{timely reachable} at round $r'$, if for all other correct processes $p_j$ in the system, $p_j$ 
	has received some message from $p_i$  within the past $\Delta_l$ interval from $r'$.\\
	Conversely, a process $p_i$ is \textbf{unreachable}, if $p_i$ is not timely reachable at some round by at least one other correct process.
\end{definition}

\begin{definition}\label{pfrd}
	A Proactive Reachability Failure Detector ($\mathcal{PR}$) guarantees the following properties:
	\begin{enumerate}
		\item \textbf{Strong timed completeness:} Every unreachable process is permanently suspected by all correct processes, at most $\Delta_m$ rounds after becoming unreachable.
		\item  \textbf{Strong timed accuracy:} No correct process that is timely reachable is suspected before crashing or becoming unreachable. 
	\end{enumerate} 
\end{definition}
\vspace*{-5pt}

Note that the notion of reachability subsumes the notion of crashed (faulty). It is also important to observe that the notion of reachability in Definition~\ref{reachable} is the one that restricts the set of correct and reachable processes the most: the set of processes termed as unreachable is largest under Definition~\ref{reachable}. In this sense, impossibilities proven with $\mathcal{PR}$ using our definition of reachability are the strongest.




An example of a $\mathcal{PR}$ implementation in our system is one that: (1) \textcolor{red}{relies on timeouts (e.g., expiring every $\Delta_l$ rounds) and periodic message (heartbeat) exchange between processes}, (2) defines timeouts on process $p_i$ as a multiple of the period at which $p_i$ sends heartbeats, and (3) ``suspects'' a process $p_i$ only when some correct process cannot receive, within the specified timeouts, (direct or indirect) heartbeats sent by~$p_i$.

\begin{assumption}\label{sysob}
	 Detector $\mathcal{PR}$ is oblivious of algorithm $\mathcal{A}$ that implements RTBRB.
\end{assumption}
\begin{assumption}\label{monito}
	Algorithm $\mathcal{A}$ uses $\mathcal{PR}$.
\end{assumption}

\begin{assumption}\label{nodep}
	Algorithm $\mathcal{A}$ does not use any additional software or hardware components besides $\mathcal{PR}$ and the network.
\end{assumption}
\begin{theorem}\label{theorem imp oblivious monitoring}
	No algorithm $\mathcal{A}$  can implement the RTBRB abstraction under Assumptions~\ref{sysob},~\ref{monito} and~\ref{nodep}.
\end{theorem}
	Roughly speaking, the proof hinges on two main ideas:
	\begin{enumerate}
		\item The reachability uncovered by $\mathcal{PR}$ has a different time-window than that of $\mathcal{A}$. \textcolor{red}{Namely, suspected processes (by $\mathcal{PR}$) are processes that are not reachable up to the current moment, while $\mathcal{A}$ would need to see the reachability from the current moment onward.}
		\item In the same network, due to probabilistic violations of reliability and timeliness, different oblivious applications (which send messages independently of each other) can reach different processes given specific time-constraints. Hence, processes considered to be reachable from $\mathcal{PR}$'s perspective might not be reachable in time by the broadcast algorithm.
	\end{enumerate}
	  \textcolor{red}{The detailed proof of Theorem~\ref{theorem imp oblivious monitoring} can be found in Appendix~A.}

\begin{corollary}
	Detecting unreachable processes using a failure detector is not sufficient to allow the implementation of the RTBRB abstraction if that failure detector is oblivious of the algorithm implementing RTBRB.
\end{corollary} 

In short, in order to implement RTBRB in a probabilistically synchronous network, we need to (1) integrate the proactive reachability failure detector within the protocol implementing RTBRB and (2) allow this compound to decide which processes are reachable within the hard timing~constraints.
In other words, in order to implement RTBRB, the reachability from the perspective of the algorithm implementing RTBRB should be reconciled with that of $\mathcal{PR}$. We accordingly propose in what follows a solution that implements RTBRB by embedding the $\mathcal{PR}$ functionality in the RTBRB protocol. 


\vspace*{-5pt}\section{An Algorithm Implementing RTBRB}\label{implementing RTBRB}

In this section, we propose an algorithm, which we call \textit{RT-ByzCast} that implements the RTBRB abstraction by avoiding the traditional \textcolor{red}{design pitfalls highlighted in Section~\ref{feasibility of implementing RTBRB}.}
\textcolor{red}{\textit{RT-ByzCast} relies on three main things: processes monitoring each other, message diffusion and signature aggregation over a time-window, and processes capable of crashing themselves. We present the intuition behind each.}
\paragraph{Process monitoring} Processes do not know beforehand when a broadcast might be invoked. \textcolor{red}{In our system, omissions on any link may be of unbounded durations (even if with a small probability). Thus, there might exist a correct process $p$ that stays unaware of a broadcast invoked by process $q$: $p$ may lose for an unbounded amount of time messages relative $q$'s broadcast.}

 In order to provide any form of real-time guarantee, it is crucial to detect the occurrence of such timing violations and react to them. \textcolor{red}{We develop the \textit{proof-of-life} function (Function~\ref{alg: proof-of-life}), which requires processes to periodically exchange heartbeats (and echo received ones) when they do not know of any broadcast. A process $p$ that does not hear enough echoes of its own heartbeats can thus suspect that it either has lost communication with many other processes or that it does not know about an invoked broadcast.}
 
\paragraph{\textcolor{red}{Message diffusion and signature aggregation}} In our system, messages can be omitted, hence re-transmissions are needed to guarantee successful delivery. \textcolor{red}{A digital signature of process $p$ on some message $m$ constitutes an unforgeable proof that $p$ processed $m$, is aware of it, and has sent (relayed) $m$. Our algorithm requires processes to aggregate, over some pre-defined time-window, the different observed signatures on a received broadcast message (including their own signature). This way a process receiving a message $m$ can know (from $m$'s aggregated signature) the identity of processes that have received $m$ (within the pre-defined time-window). Besides aggregation our algorithm requires processes to diffuse a received message $m$ (along with its aggregated signature) by re-transmitting $m$ within the pre-defined time-window multiple times and to other processes. This diffusion helps mask network omissions by exploiting multiple network paths, e.g., relaying information via intermediate processes.}

In fact, our time-window creates a slack period of ``collective re-transmissions" in which, even if a message is omitted a few times, it still has a high chance (as shown in Section~\ref{Evaluation}) of being successfully received by the end of that time-window (which can be viewed now as the new deadline). 
We denote by $\mathbb{R}$ the duration of that time-window (in rounds). \textcolor{red}{$\mathbb{R}$ is a parameter of our algorithm (known to processes) and represents the delay we tolerate for a round-trip message delay.} The larger the value of $\mathbb{R}$, the bigger the slack period, but the higher the chances of a message being received in that duration (see Section~\ref{Evaluation}). In this sense, $\mathbb{R} = f(k)$, meaning that $\mathbb{R}$ is a function of $k$ (the anticipated omission degree), $\mathbb{R} \geq 2k+2$ to account for a round-trip bound.


\begin{figure}[t] 
	\begin{center}
		\includegraphics[scale=0.22]{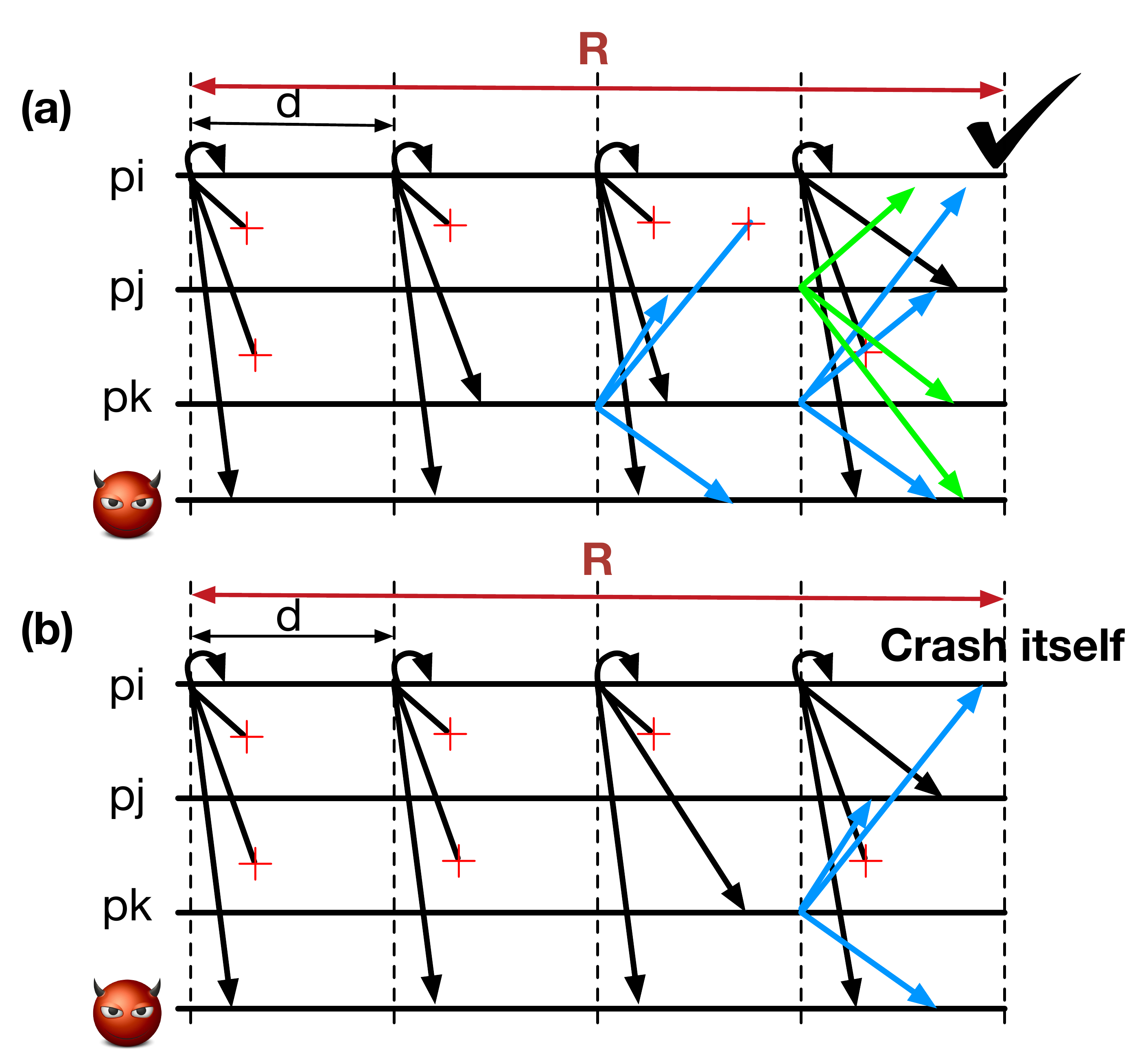}
		\begin{center}
			\vspace*{-10pt}
			\caption{ (a) $p_i$ gathers $2f+1$ signatures in $\mathbb{R}$ rounds; (b) $p_i$ does not gather $2f+1$ signatures in $\mathbb{R}$ rounds and crashes itself.}\label{exmaple exec}
			\vspace*{-20pt}
		\end{center} 
	\end{center}
\end{figure}

\paragraph{Self-crash capabilities} 
\textcolor{red}{A message sent by a non-Byzantine process $p$ might still fail to reach ``enough'' processes, despite diffusion and aggregation within a time-window of duration $\mathbb{R}$ (assumptions on omission degree, despite their high coverage, might be violated). In that case, our algorithm triggers $p$ to crash itself (we formally define when a process crashes itself in Definition~\ref{self-crash}). 
	A non-Byzantine process $p$ that ``crashes itself'' is no longer correct with respect to the algorithm's correctness properties. }
	



\floatname{algorithm}{Function}
\begin{algorithm}[t]
	\caption{ \textit{proof-of-life($\Delta_{r}$)}@ process $p_j$} \label{alg: proof-of-life}
	\begin{algorithmic}[1]
		\footnotesize	
		\For{every round $r$}
		\State \textbf{Send} \texttt{HB($\{p_j,r\};\Phi_{p_j}$)} to all $p\in \mathit{\Pi}$
		\EndFor
		\Event {$<$receive \texttt{HB($\{p_i,r\};\textcolor{red}{sigs}$)} at some round $r'>$}
		\\// $\textcolor{red}{sigs}$ designates the signatures on $\{p_i,r\}$.
		\State \textbf{Verify} message validity
		\State \textbf{Set} $\textcolor{red}{sigs} =\textcolor{red}{sigs} \bigcup \Phi_{p_j}$
		\For {every round in $[r'+1,r+\Delta_{r}]$}
		\State \textbf{Send} \texttt{HB($\{p_i,r\};\textcolor{red}{sigs}$)} to all $p\in \mathit{\Pi}$
		\EndFor 
		\EndEvent 
		\For {every round $r$ where $r>time_{w}$}
		\State \textbf{Compute} $\mathcal{R}_{HB}(p_j,r)$
		\\// $\mathcal{R}_{HB}(p_j,r)$:set of processes that received some heartbeat from $p_j$ (directly or indirectly), during some round in $[r-1-\Delta_{r},r~-~1]$
		\If{$|\mathcal{R}_{HB}(p_j,r)|\leq~2f$}
		\State \textbf{Transition} to \textit{Dead-State}
		\EndIf
		\EndFor 
	\end{algorithmic}
\end{algorithm}

\vspace*{-30pt}
\subsection{RT-ByzCast Overview}
We present now a high-level view of how our algorithm works in systems where $n=3f+1$ processes. 
\subsubsection{Behavior of Processes Unaware of a Broadcast}
\textcolor{red}{Every process that did not receive any message relative to a broadcast, executes the \textit{proof-of-life} function (Function~\ref{alg: proof-of-life}), where a process periodically sends a signed heartbeat message to all other processes.  When process $p_j$ receives $p_i$'s heartbeat at round $r$, $p_j$: (i) aggregates signatures relative to $p_i$'s heartbeats (signatures of heartbeats sent by $p_i$ in $[r,r-\mathbb{R}]$), (ii) appends its own signature to the formed aggregate, and (iii) periodically echoes $p_i$'s heartbeat, with the aggregated signatures,} to all other processes (see Figure~\ref{exmaple exec}). A process $p$ executing the \textit{proof-of-life} crashes itself when $p$ does not receive in any time-window \textcolor{red}{of duration $\mathbb{R}$}, $2f+1$ signatures on its own heartbeats (including its own signature).

\subsubsection{Behavior of Processes Knowledgeable of a Broadcast}
In this case, a process echoes periodically the broadcasted value, only if that value is indeed sent by the process that originally issued the broadcast. If multiple values are heard relative to a single process, the first heard value is the one to be echoed. \textcolor{red}{Processes continue to execute the proof-of-life function, however, piggybacking proof-of-life messages to the echo messages of the broadcast.}

Similar to the proof-of-life function, 
processes aggregate and append signatures relative to a given broadcast value in a window of $\mathbb{R}$ rounds (from the current round). \textcolor{red}{Upon receiving a value that is signed by more than $2f$ processes, a process delivers that value. A process that does not receive more than $2f$ signatures in a $\mathbb{R}$ time-window after being aware of the broadcast crashes itself (see Definition~\ref{self-crash}).}

\subsubsection{RT-ByzCast Latency}	
Our \textit{RT-ByzCast} algorithm delivers messages within an upper bound of $3\mathbb{R}$ rounds. The following example briefly highlights the intuition behind requiring $3\mathbb{R}$ (the detailed proof can be found in Appendix~B. 
Assume a correct process issues a broadcast at round $r$. \textcolor{red}{By round $r+\mathbb{R}$, at least $2f$ other} processes would have seen that broadcast (otherwise the broadcasting process kills itself). In the worst case, those processes would see the broadcast at round $r+\mathbb{R}-1$. Assume that the remaining $f$ processes in \textcolor{red}{the system have not yet received the broadcast message}. By round $r+2\mathbb{R}$ these $f$ processes should have heard about the broadcast or otherwise they would kill themselves (processes that already received the broadcast send heartbeats only via echo messages relative to the broadcast). After receiving the broadcast, these $f$ processes have an additional $\mathbb{R}$ rounds to collect enough signatures and hence will deliver the broadcast message at worst by $r+3\mathbb{R}$.
\floatname{algorithm}{Function}
\begin{algorithm}[t]
	\caption{ \textit{aggregate-sig$_{p_j}$($v, p_i, \Phi_{p_x},...,\Phi_{p_z}, p_k$)}} \label{aggregate functions}
	\begin{algorithmic}[1]
		\footnotesize	\If {$(v,...) \notin \texttt{Msg}[p_i][p_k]$} 
		\State \textbf{Set} $\textcolor{red}{sigs} = \Phi_{p_j}\bigcup\Phi_{p_x},...,\Phi_{p_z}$,  after verifying that $p_x,...,p_z$ indeed signed~$v$.
		\State \textbf{Add} $(v,\textcolor{red}{sigs})$ to $\texttt{Msg}[p_i][p_k]$,
		\EndIf
		\If {$(v,...)\in \texttt{Msg}[p_i][p_k]$} 
		\State Consider that $\textcolor{red}{sigs}=...$, then \textbf{update} $(v,\textcolor{red}{\textcolor{red}{sigs}})$ in $\texttt{Msg}[p_i][p_k]$ such that $(v,\textcolor{red}{sigs})$ includes as well the signatures relative to $p_x,...,p_z$, i.e., $\textcolor{red}{sigs}=\textcolor{red}{sigs}\bigcup{\Phi_{p_x},...,\Phi_{p_z}}$.
		\EndIf
		\State \textbf{return} $\textcolor{red}{sigs}$;
	\end{algorithmic}
\end{algorithm}
\begin{algorithm}[t]
	\caption{ \textit{deliver-message$_{p_j}$($p_i, v,\textcolor{red}{sigs}$)} @ round $r$} \label{deliver functions}
	\begin{algorithmic}[1]
		\footnotesize
		\State \textbf{Deliver} $v$ (if not already delivered) 
		\State \textbf{Initialize} $\mathcal{R}_{deliver}(p_i,r)=p_j$ 
		\State \textbf{Stop} sending any \texttt{Echo()}
		
		\State \textbf{Send} \texttt{Deliver$_{p_j}$($(p_i,v,\textcolor{red}{sigs});\Phi_{p_j}$)} at the beginning of every round in $[r+1,r+1+2\mathbb{R}]$, where $\textcolor{red}{sigs}$ contains signatures of all processes in $\mathcal{R}_{echo}(p_i,r_0,v)$. 
		
		// $\Phi_{p_j}$ is $p_j$'s signature on the new payload: $(p_i,v,\textcolor{red}{sigs})$
		\State At the beginning of round $r+2+2\mathbb{R}$:
		\If {(no \texttt{Echo()}  or \texttt{RTBRB-broadcast()} is being sent)} 
		\State \textbf{Execute} \textit{proof-of-life($\mathbb{R}$)} (non-piggyback mode).  
		
		\EndIf
	\end{algorithmic}
\end{algorithm}

\vspace*{-20pt}\subsection{\textcolor{red}{Description of RT-ByzCast}}\label{RT-ByzCast}
\textcolor{red}{We now describe our algorithm formally and in more details.
Using our \textit{RT-ByzCast } algorithm, non-Byzantine processes can be in one of two states: the \textit{Alive-State} or the \textit{Dead-State}. Non-Byzantine processes in the \textit{Alive-State} follow the algorithm faithfully}. Processes in the \textit{Dead-State}, however, refrain from sending any messages and do not respond to any higher-level application requests. A process in the \textit{Dead-State} appears to have crashed as in \textit{fail-stop} failures~\cite{Schneider:1984:BGA:190.357399}. For presentation simplicity, we consider that once a process is in the \textit{Dead-State}, it remains in the \textit{Dead-State} forever. \textcolor{red}{In Section~\ref{Dead-State Revival}, we discuss how processes can be revived, by returning to the \textit{Alive-State} and participating in the algorithm}.

Initially, after receiving an initiation signal from the clock, all processes are awakened and execute in the \textit{Alive-State}. 
Every process $p_i$ defines a \textcolor{red}{time-window} of a known fixed size of $\mathbb{R}>1$ rounds and $\texttt{Msg}[]^n[]^n$, an initially empty array to store the values broadcast relative to each process, with the corresponding  signatures on these values. $\texttt{Msg}[]^n[]^n$ is the local data structure that stores relative to each process that issues a broadcast and each process that echoes it, the broadcast values and their aggregated signatures. For example, $\texttt{Msg}[p_i][p_k]$ stores all messages with aggregated signatures that are broadcast by some process $p_i$ and that are echoed by a process~$p_k$. After initialization, every process in the \textit{Alive-State} executes a \textit{proof-of-life} \textcolor{red}{function with parameter~$\Delta_{r}=\mathbb{R}$}. 

A process $p_i$ that executes the proof-of-life function sends a heartbeat \texttt{HB($\{p_i,r\};\Phi_{p_i}$)} to all processes at the beginning of every round~$r$, where $\Phi_{p_i}$ denotes $p_i$'s signature over the heartbeat. A process $p_j$ receiving \texttt{HB($\{p_i,r\};\textcolor{red}{sigs}$)} at some round $r'$ verifies that the heartbeat has been indeed seen by the processes that signed it ($\textcolor{red}{sigs}$ is the set of signatures on $\{p_i,r\}$). Then $p_j$ appends its signature to $\textcolor{red}{sigs}$ and sends \texttt{HB($\{p_i,r\};\textcolor{red}{sigs}$)} to all other processes in all rounds $\in[r'+1,r+\Delta_{r}]$ (lines 7-10). A process $p_i$ computes at the beginning of every round $r:r>\Delta_{r}$ the set $\mathcal{R}_{HB}(p_i,r)$, which is the set of distinct processes that were able to receive some heartbeat from $p_i$ (directly or indirectly), during some round in $[r-1-\Delta_{r},r~-~1]$.

Precisely, let us denote by $M$ the set of \texttt{HB($\{p_i,r''\};*$)}  heartbeats that $p_i$ received. Then $\mathcal{R}_{HB}(p_i,r)$ contains all processes that have their signature in any heartbeat in $M$ such that round $r''\in [r-1-\Delta_{r},r-1]$.  For any round after $\Delta_{r}$ rounds elapse since system initialization (i.e., $\forall$~rounds $r:r>\Delta_{r}$), if ever $|\mathcal{R}_{HB}(p_i,r)|\leq~2f$, then process $p_i$ transitions to the \textit{Dead-State} (lines 12-18).

 \vspace{5pt}
 \subsubsection*{\textit{RT-ByzCast: The Algorithm} (Algorithm~\ref{main alg})} \label{main algo}
 After system start-up, every process $p_i$ begins by executing the proof-of-life function with $\Delta_{r}= \mathbb{R}$. A process $p_i$ that wishes to broadcast a value $v$ at round $r-1$ 
 initializes an empty set denoted by $\mathcal{R}_{echo}(p_i,r,v)$. Then $p_i$ signs $(p_i,r-1,v)$ with an unforgeable signature $\Phi_{p_i}$ and produces the tuple $\{(p_i, r-1,v);\Phi_{p_i}\}$. This tuple is made of two parts: a \textit{payload} and a \textit{signed-by} part. For example, the tuple $\{(p_i,r-1,v);\Phi_{p_i}\}$ has $(p_i,r-1,v)$ as payload and $\Phi_{p_i}$ as the signed-by.  \textcolor{red}{Afterwards process $p_i$ sends at the beginning of round $r-1$, the message \texttt{RTBRB-broadcast($(p_i,r-1,v);\Phi_{p_i}$)} to all processes. 
 When a process $p_j$ (that is executing the proof-of-life function) first receives an \texttt{RTBRB-broadcast($(p_i,r',v);\Phi_{p_i}$)} message in round $r\geq r'$ from $p_i$}, process $p_j$ initializes an empty set denoted by $\mathcal{R}^{p_j}_{echo}(p_i,r+1,v;\Phi_{p_i})$. \textcolor{red}{Afterwards} $p_j$ retrieves the tuple $ \{(p_i,r',v);\Phi_{p_i}\}$, verifies that $p_i$ indeed sent $v$, and produces the tuple $\{(p_i,r',v;\Phi_{p_i});\Phi_{p_j}\}$ by appending $p_j$'s signature to the signed-by field.
 
\textcolor{red}{Process $p_j$  starts sending at the beginning of every round (as of round $r$ onward) \texttt{Echo$_{p_j}$($(p_i,r',v;\Phi_{p_i});\Phi_{p_j}$)} to all other~processes.} Process $p_j$ continues to execute the \textit{proof-of-life} function however by piggybacking its heartbeats on \texttt{Echo$_{p_j}$($(p_i,r',v;\Phi_{p_i});\Phi_{p_j}$)}.

\textcolor{red}{When process $p_j$ receives a (valid) \texttt{Echo$_{p_k}$($(p_i,r',v;\Phi_{p_i});\Phi_{p_x},...,\Phi_{p_z}$)} at some round $r$ , $p_j$ behaves as follows.} 

\textbf{A.~}If $p_j$ is not sending any \texttt{Echo()}(lines 14-27):
then if $p_j$ has not already delivered a message from~$p_i$, $p_j$ aggregates signatures in the variable $\textcolor{red}{sigs}$. 
It does so using the function \textit{aggregate-sig$_{p_j}$($...$)} (detailed Function~\ref{aggregate functions}) that combines all observed signatures on messages sent by $p_k$ containing $(p_i,r',v)$. 
Then if signatures in $\textcolor{red}{sigs}$ do not add up to $2f+1$ processes, $p_j$ sends \texttt{Echo$_{p_j}$($(p_i,r+1,v;\Phi_{p_i});\Phi_{p_j}$)} and initializes a set ($\mathcal{R}_{echo}(p_i,r+1,v)$) that tracks the number of signatures collected on the \texttt{Echo$_{p_j}$($(p_i,r+1,v;\Phi_{p_i});...$)}. Otherwise, i.e., if signatures in $\textcolor{red}{sigs}$ amount to at least $2f+1$ processes, $p_j$ sets $\mathcal{R}_{echo}(p_i,r,v)=\textcolor{red}{sigs}$ and  delivers the message by executing \textit{deliver-message($v,\textcolor{red}{sigs}$)}.	Any process $p_j$ that executes \textit{deliver-message($p_i,v,\textcolor{red}{sigs}$)} delivers message $(p_i,v,\textcolor{red}{sigs})$ and only sends \texttt{Deliver$_{p_j}$($(p_i,v,\textcolor{red}{sigs});...$)} for $2\mathbb{R}$ rounds. $p_j$ should collect a least $2f+1$ signatures during the first $\mathbb{R}$ rounds. This requirement ensures that there is a full $\mathbb{R}$-duration during which at least $f+1$ correct processes are sending only such deliver messages\footnote{This requirement ensures agreement on delivering the same message among all correct processes, when broadcasting nodes are malicious.}. 

If no \texttt{Echo()} or \texttt{RTBRB-broadcast()} is being sent/received, $p_j$ executes the proof-of-life function with parameter $\mathbb{R}$ (without piggybacking). 

\floatname{algorithm}{Algorithm}
\begin{algorithm}[H] 
	\caption{\textit{RT-ByzCast}} \label{main alg}
	\begin{algorithmic}[1]
		\footnotesize	
		\State \textbf{Init:} $\texttt{Msg}[]^n[]^n = \emptyset$ 
		\State \textbf{Execute} \textit{proof-of-life($\mathbb{R}$)};
		\Event {$<$$p_i$ wants to broadcast a value $v$$>$}
		\State \textbf{Execute} \textit{proof-of-life} function in piggyback mode
		\State \textbf{Initialize} $\mathcal{R}_{echo}(p_i,r,v)=\emptyset$ 
		\State \textbf{Send} periodically starting from the current round \hspace*{4.2mm}\texttt{RTBRB-broadcast($(p_i,r_{current},v);\Phi_{p_i}$)} to all $p\in \mathit{\Pi}$
		\EndEvent
		
		\Event{$<$receive \texttt{RTBRB-broadcast($(p_i,r',v);\Phi_{p_i}$)}in round $r-1\geq r'$ for the first time$>$}
		\State \textbf{Execute} \textit{proof-of-life} function in piggyback mode
		\State \textbf{Initialize} $\mathcal{R}^{p_j}_{echo}(p_i,r,v)=p_i;$ $\textcolor{red}{sigs} = \Phi_{p_j}$
		\State \textbf{Send} \texttt{Echo$_{p_j}$($(p_i,r',v;\Phi_{p_i});\textcolor{red}{sigs}$)} to all $p\in \mathit{\Pi}$ at rounds$\geq r$
		\EndEvent
		
		\hspace{-12mm}\textbf{$@$ process $p_j$:}
		\Event{$<$receive \texttt{Echo$_{p_k}$($(p_i,r',v;\Phi_{p_i});\Phi_{p_x},...,\Phi_{p_z}$)} at round~$r$}
		
		\If{$p_j$ is not sending any \texttt{Echo()}}
		\State \textbf{Set} $\textcolor{red}{sigs}$ = \textit{aggregate-sig$_{p_j}$($v, p_i, \Phi_{p_x}...\Phi_{p_z}, p_k$)}
		\If {$p_j$ has not already delivered a message relative to~$p_i$}
		\State \textbf{Execute} \textit{proof-of-life} function in piggyback mode. 
		\State \textbf{Initialize} $\mathcal{R}_{echo}(p_i,r+1,v)=\emptyset$
		\If {$\textcolor{red}{sigs}\leq 2f$} 
		\State \textbf{Send} at the beginning of every cycle (as of round $r+1$ \hspace*{13mm}onward)  \texttt{Echo$_{p_j}$($(p_i,r+1,v;\Phi_{p_i});\Phi_{p_j}$)} if $k\neq j$
		\EndIf
		\EndIf
		\If {$\textcolor{red}{sigs}> 2f$ (for the first time)} 
		\State \textbf{Set} $\mathcal{R}_{echo}(p_i,r,v)=\textcolor{red}{sigs}$ 
		\State \textbf{Execute} \textit{deliver-message($p_i,v,\textcolor{red}{sigs}$)}
		\EndIf
		\EndIf
		
		\If {$p_j$ is sending an \texttt{Echo$_{p_j}$($(p_i,r',v;\Phi_{p_i});...$)}}
		\State \textbf{Set} $\textcolor{red}{sigs}$ = \textit{aggregate-sig$_{p_j}$($v, p_i, \Phi_{p_x}...\Phi_{p_z}, p_k$)}
		\If {$\textcolor{red}{sigs}> 2f$ (for the first time)}
		\State Set $\mathcal{R}_{echo}(p_i,r,v)=\textcolor{red}{sigs}$
		\State \textbf{Execute} \textit{deliver-message($p_i,v,\textcolor{red}{sigs}$)}.
		\EndIf 
		\If {$\textcolor{red}{sigs}\leq 2f$ $\land$ $k=j$} 
		\State \textbf{Set} $\mathcal{R}_{echo}(p_i,r,v)=\textcolor{red}{sigs}$. 
		\EndIf
		\If {$\textcolor{red}{sigs}\leq 2f$}
		\State \textbf{Send} at the beginning of every cycle (as of round $r+1$ \hspace*{13mm}onward) \texttt{Echo$_{p_k}$($(p_i,r',v;\Phi_{p_i});\textcolor{red}{sigs}$)} to all $p \in \mathit{\Pi}$
		\EndIf
		\EndIf
		
		\If {$p_j$ is sending \texttt{Echo$_{p_j}$($(p_i,r'',v';\Phi_{p_i});*$)} $:v'\neq v$}
		\State \textbf{Set} $\textcolor{red}{sigs}$ = \textit{aggregate-sig$_{p_j}$($v', p_i, \Phi_{p_x}...\Phi_{p_z}, p_k$)}
		\If {$\textcolor{red}{sigs}> 2f$ (for the first time)}
		\State \textbf{Set} $\mathcal{R}_{echo}(p_i,r,v)=\textcolor{red}{sigs}$
		\State \textbf{Execute} \textit{deliver-message($p_i,v',\textcolor{red}{sigs}$)}
		\EndIf
		\EndIf
		\EndEvent 
		\Event {\textcolor{red}{$<$receive \texttt{Deliver$_{p_k}$($(p_i,v,\textcolor{red}{sigs});\Phi_{p_x}...\Phi_{p_z}$)} at round~$r$$>$}}
		\If {\textcolor{red}{$((p_i,v)$ is not delivered yet}}
		\State \textcolor{red}{\textbf{Deliver} $v$}
		\State \textbf{Stop} sending any \texttt{Echo()} 
		\State \textcolor{red}{\textbf{Initialize} set $\mathcal{R}_{deliver}(p_i,r)= \{p_x,...,p_z\}$}
		\Else
		\State \textcolor{red}{$\mathcal{R}_{deliver}(p_i,r')=\mathcal{R}_{deliver}(p_i,r')\bigcup \{p_x,...,p_z\}$.} 
		\EndIf
		\State \textbf{Send} \texttt{Deliver$_{p_j}$($(p_i,v,\textcolor{red}{sigs});signatures$)} to all $p\in\mathit{\Pi}$~at \hspace*{2.8mm} every round in $[r+1,r+1+2\mathbb{R}]$, $signatures$  contains the signatures \hspace*{4mm}of all processes in $\mathcal{R}_{deliver}(p_i,...)$. 
		\State \textbf{Execute} same commands as lines 5-8 of \textit{deliver-message}$_{p_j}$(...)
		\EndEvent 
	\end{algorithmic}
\end{algorithm}

\noindent\textbf{B.~} If $p_j$ is sending an \texttt{Echo$_{p_j}$($(p_i,r',v;\Phi_{p_i});...$)} \textcolor{red}{(lines 30-42)}: \textcolor{red}{In this case, if $p_j$ has not delivered a message from~$p_i$, $p_j$ aggregates signatures in $\textcolor{red}{sigs}$ (using the function \textit{aggregate-sig$_{p_j}$($...$)})}. Then $p_j$ checks if $\textcolor{red}{sigs}$ contains more than $2f$ signatures, in which case $p_j$ sets $\mathcal{R}_{echo}(p_i,r,v)=\textcolor{red}{sigs}$ and delivers $(p_i,r',v)$ by executing \textit{deliver-message($p_i,v,\textcolor{red}{sigs}$)}. Otherwise, i.e., if $|\textcolor{red}{sigs}|\leq 2f$, $p_j$ sends echoes of $(p_i,r',v)$ with the new aggregated signatures, if ($k=j$).

\noindent\textbf{C.~}If $p_j$ is sending an \texttt{Echo$_{p_j}$($(p_i,r'',v';\Phi_{p_i});*$)} where $v'\neq v$: $p_j$ aggregates signatures in $\textcolor{red}{sigs}$ (using the function \textit{aggregate-sig$_{p_j}$($v', p_i, \Phi_{p_x}...\Phi_{p_z}, \Phi_{p_k}$)}). Then $p_j$ checks if $\textcolor{red}{sigs}$ contains more than $2f$ signatures, in which case $p_j$ delivers $(p_i,r'',v')$ by executing \textit{deliver-message($p_i,v',\textcolor{red}{sigs}$)}.

If $p_j$ receives \texttt{Deliver$_{p_k}$($(p_i,v,\textcolor{red}{sigs});\Phi_{p_x}...\Phi_{p_z}$)} for the first time at round~$r$ (lines 49-58), $p_j$ delivers $(p_i,v,\textcolor{red}{sigs})$ by executing \textit{deliver-message($p_i,v,\textcolor{red}{sigs}$)}.
If $p_j$ has previously seen such a \texttt{Deliver()} message, $p_j$ aggregates in the variable $signatures$ all the signatures it has on \texttt{Deliver$_{p_j}$($(p_i,v,\textcolor{red}{sigs});...$)} messages. Then $p_j$ sends \texttt{Deliver$_{p_j}$($(p_i,v,\textcolor{red}{sigs});signatures$)} at the beginning of every round $\in[r+1,r+1+2\mathbb{R}]$. 

\textcolor{red}{At the beginning of round $r+2+2\mathbb{R}$, if $p_j$ is not sending any \texttt{Echo()} or \texttt{RTBRB-broadcast()}, then $p_j$ executes the proof-of-life function with parameter $\mathbb{R}$ (without piggybacking).} 

\textcolor{red}{\vspace*{-5pt}\begin{definition}\label{def: invalid}
	Any \texttt{RTBRB-broadcast($p_i,...$)} or \texttt{Echo($p_i,...$)} message is termed \textit{invalid} when it possesses any incorrect signature. 
	Any \texttt{Deliver($(p_i,...,\textcolor{red}{sigs})...$)} is termed \textit{invalid} when it possesses any incorrect signature, 
	or when $\textcolor{red}{sigs}$ has less than $2f+1$ correct signatures. If a process receives invalid messages, it ignores them.
\end{definition}}
\textcolor{red}{\vspace*{-5pt}\begin{remark}\label{rem: piggyback}
	A process sending \texttt{RTBRB-broadcast()}, \texttt{Echo()} or \texttt{Deliver$_{p_j}$()}, executes \textit{proof-of-life} by piggybacking heartbeats to the sent \texttt{RTBRB-broadcast()}, \texttt{Echo$_{p_j}$()} and \texttt{Deliver$_{p_j}$()}; otherwise heartbeats are sent in individual messages. This is represented in lines (4), (9), and (17) of Algorithm~\ref{main alg}.
\end{remark}}
  
\vspace*{-5pt}\begin{definition}\label{self-crash}
\textcolor{red}{	A process $p_j$ crashes itself and transitions to the \textit{Dead-State} at the end of round $r_{cur}$ if any of the four cases below is satisfied.}
	
	\textbf{Case 1.} A process $p_j$ transitions to the \textit{Dead-State} if both conditions below hold.
	\begin{enumerate}
		\item \textcolor{red}{$\exists r': |\mathcal{R}_{echo}(p_i,r',v)|\leq 2f,~\forall v$ $\land$ $r_{cur} \geq r'+\mathbb{R}$}
		\item If $p_j$ has not \textit{discovered a }\textit{lie}. A process $p_j$ is said to ``discover a lie'', if $p_j$ can verify that two different values have been sent by the process issuing the broadcast. 
	\end{enumerate}

	\indent \textbf{Case 2.} $p_j$ transitions to the \textit{Dead-State} if both conditions hold: $\exists r': |\mathcal{R}_{deliver}(p_i,r')|\leq 2f$, and $r_{cur} \geq r'+\mathbb{R}$.
	
	\textbf{Case 3.} Process $p_j$ transitions to the \textit{Dead-State} if $p_j$ does not hear from at least $2f+1$ processes in any $\mathbb{R}$ window.
	
	\textbf{Case 4.} \textcolor{red}{Process $p_j$ transitions to the \textit{Dead-State} if $p_j$ in some duration $\mathbb{R}$ does not see echoes of its messages from at least $2f+1$ processes. }
\end{definition}

\vspace*{-5pt}\textcolor{red}{We give now a brief intuition behind each of the four cases. In Case 1,  the first condition ensures that processes, which cannot communicate with a Byzantine quorum, crash themselves. 
The second condition avoids the denial of service attack that a Byzantine process can launch. 
For Case 2 and Case 3, consider a setting where a correct process $p$ aggregates $2f+1$ signatures, delivers a message and directly loses communication with other processes.} Now assume that other processes detect that the broadcasting process is lying. These other processes do not kill themselves (by Case 1). If $p$ stays alive this would violate agreement. Case 2 eliminates this~violation.  Case 3 eliminates a violation in agreement when the processes that detect a lie are fewer than $f+1$ and get disconnected temporarily from the other correct processes that deliver some~$v$. 
Case 4 makes sure that any process whose messages are not seen by (at least) $2f+1$ processes in any duration of value $\mathbb{R}$ crashes itself. This case 4 is used to allow the detection of processes that crash themselves (see Section~\ref{detect dead-state}).

\textcolor{red}{We present in Appendix~B 
the proof that \textit{RT-ByzCast} implements the RTBRB abstraction.}

\vspace*{-10pt}\section{Evaluation}\label{Evaluation}

\textcolor{red}{We now present an analytic performance study of \textit{RT-ByzCast}. We evaluate first the probability of having a correct process in \textit{RT-ByzCast}, enter the \textit{Dead-State}, i.e., crash itself.} 

\vspace*{-10pt}\subsection{Probability of a Correct Process Shutdown}\label{process shutdown}
The probability of a correct process crashing itself is crucial, as it may hinder the reliability of the whole system, namely the system's liveness. For instance, in systems where $n=3f+1$ (which are a subset of systems with $f<\frac{n}{3}$), having a single correct process crash itself can result in a cascade of self-induced crashes leading the whole system to shutdown. This cascade is caused by the fact that one correct process crashing itself may prevent correct processes ever gathering the required $2f+1$ quorum, e.g., if in the worst case all $f$ processes are Byzantine in addition to the correct process that crashed itself. 



In our simulation setting, we first assume that all messages can be lost/omitted independently and with the same probability, i.e., we assume that $P^{o}_{ij}(t)=p~\forall~ t,i\neq j$. \textcolor{red}{In Appendix~C, we study the effect of having correlated losses/omissions that may result in bursts and show that a large enough slack window masks losses/omissions regardless of correlations and preserves drawn conclusions under the independent assumption}. We describe the simulation setting in which we perform our evaluation. 




Let us denote by $\mathcal{C}$, the minimum set of non-Byzantine processes, i.e,  $|\mathcal{C}|=n-f$. We run our simulations for $|\mathcal{C}|~\in~\{5,10,20,30,40,50,100,200\}$ processes. For each value of $|\mathcal{C}|$, we consider various probability values with which a sent message can be lost/omitted. 
 Namely, we consider the probability of losing/omitting a message sent at any point in time to be $p~\in~\{10^{-6},10^{-5},10^{-4},10^{-3},10^{-2},10^{-1}, 0.2, 0.3,0.4,0.5,\\0.6,0.7,0.8,0.9\}$. For a given value of $|\mathcal{C}|$ and $p$, we invoke a broadcast at one of the processes and record, after $\mathbb{R}$ rounds of communication, if any process does not receive $|\mathcal{C}|$ signatures on the value being broadcast. We repeat such an instance $10^6$~times. We~report our results showing: $$\resizebox{0.95\hsize}{!}{$\frac{\text{num. of instances in which some correct process crashes itself}}{10^6}$},$$ for $\mathbb{R} \in~\{5,6,10,15\}$ rounds respectively. We select these values of $\mathbb{R}$ to show that a well chosen fixed value regardless of system size and loss probability allows our \textit{RT-ByzCast} algorithm to be reliable. 

\begin{figure}[t] 
	\begin{center}
		\includegraphics[scale=0.2]{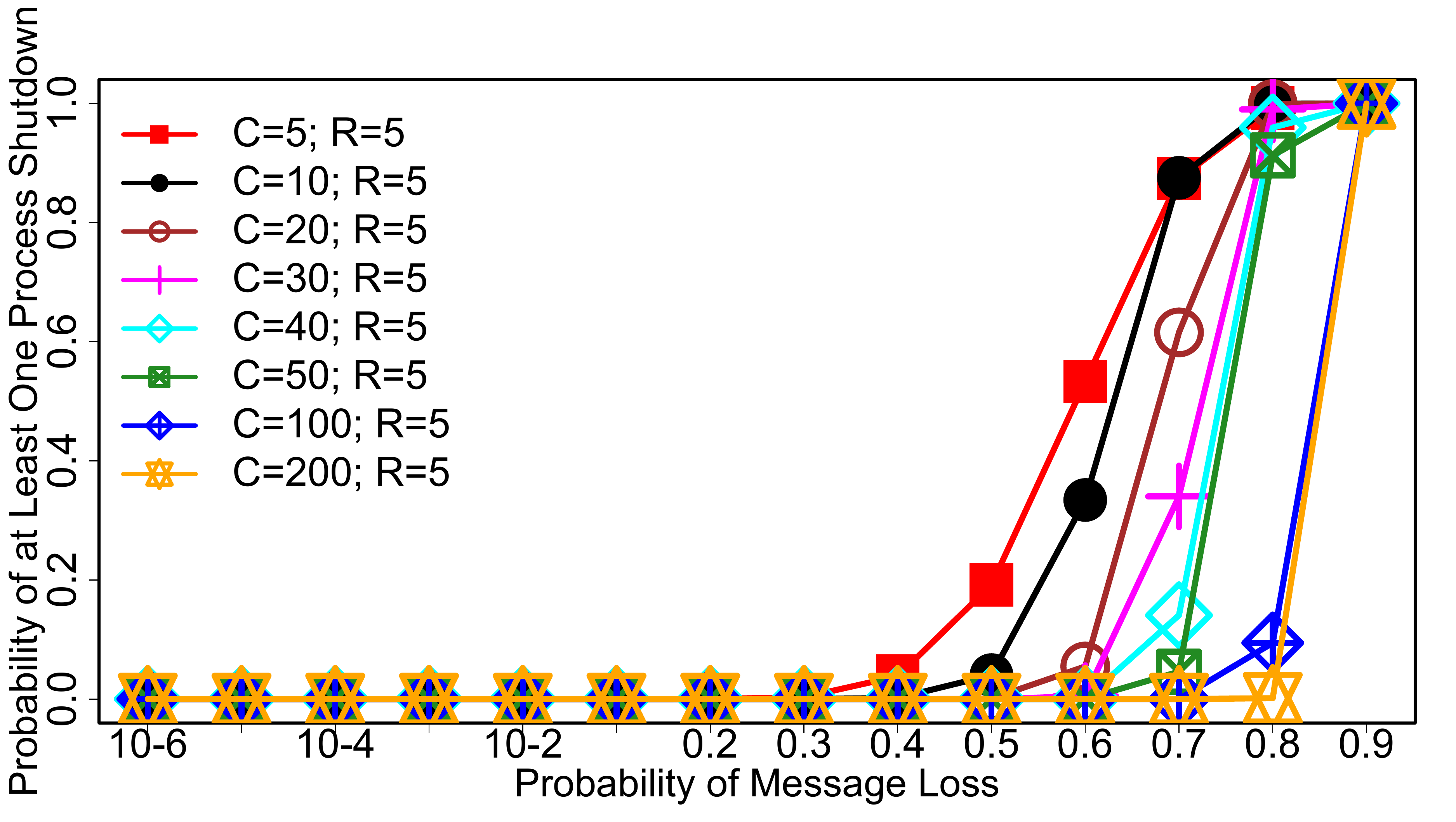}
		\begin{center}
			\vspace*{-10pt}
			\caption{The probability of a correct process crashing itself in a system after $10$ communication rounds ($\mathbb{R}=5$).}\label{pshut5}
			\vspace*{-15pt}
		\end{center} 
	\end{center}
\end{figure}
\begin{figure}[t] 
	\begin{center}
		\includegraphics[scale=0.2]{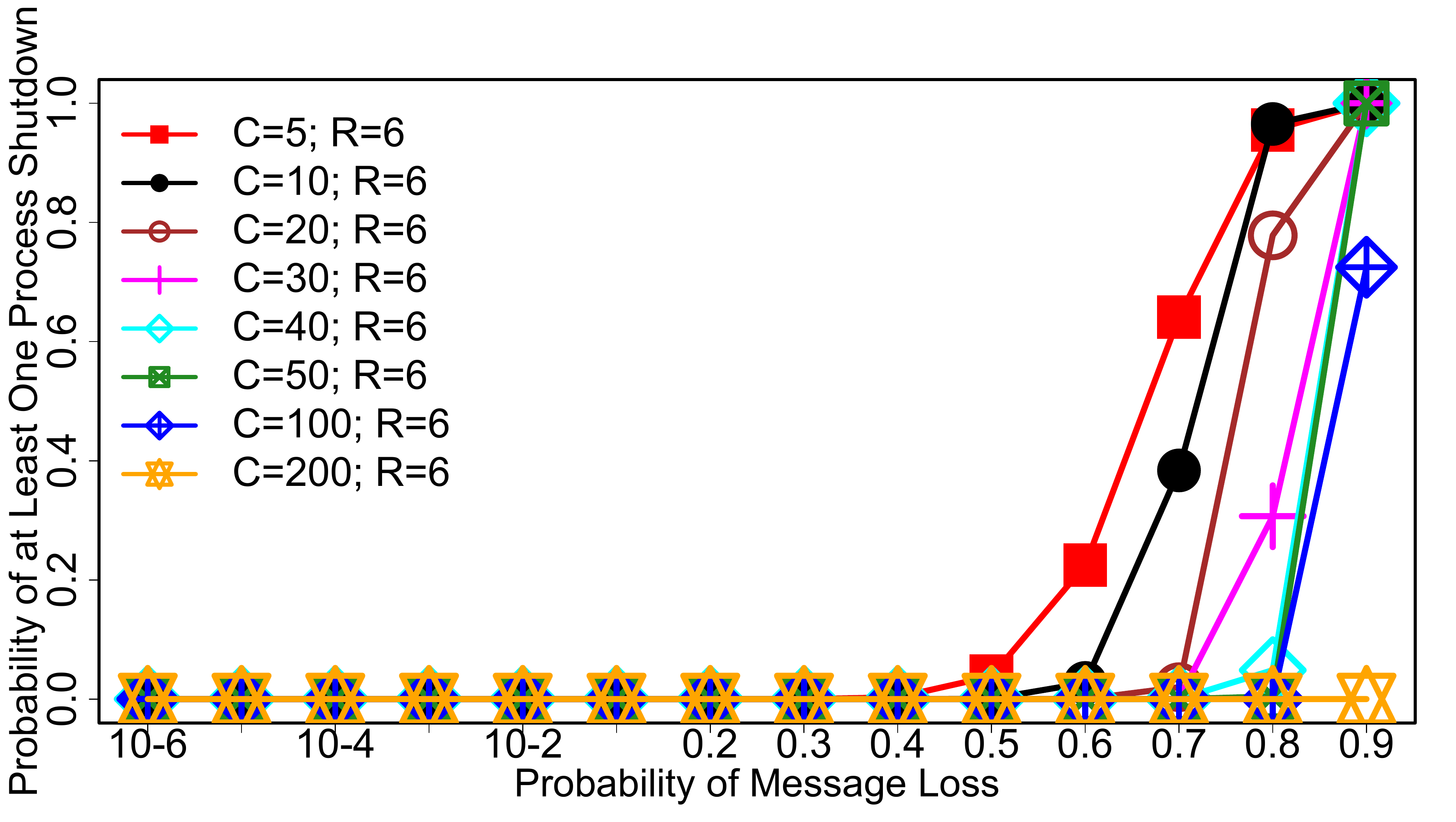}
		\begin{center}
			\vspace*{-10pt}
			\caption{The probability of a correct process crashing itself in a system after $10$ communication rounds ($\mathbb{R}=6$).}\label{pshut6}
			\vspace*{-15pt}
		\end{center} 
	\end{center}
\end{figure}
\begin{figure}[t] 
	\begin{center}
		\includegraphics[scale=0.2]{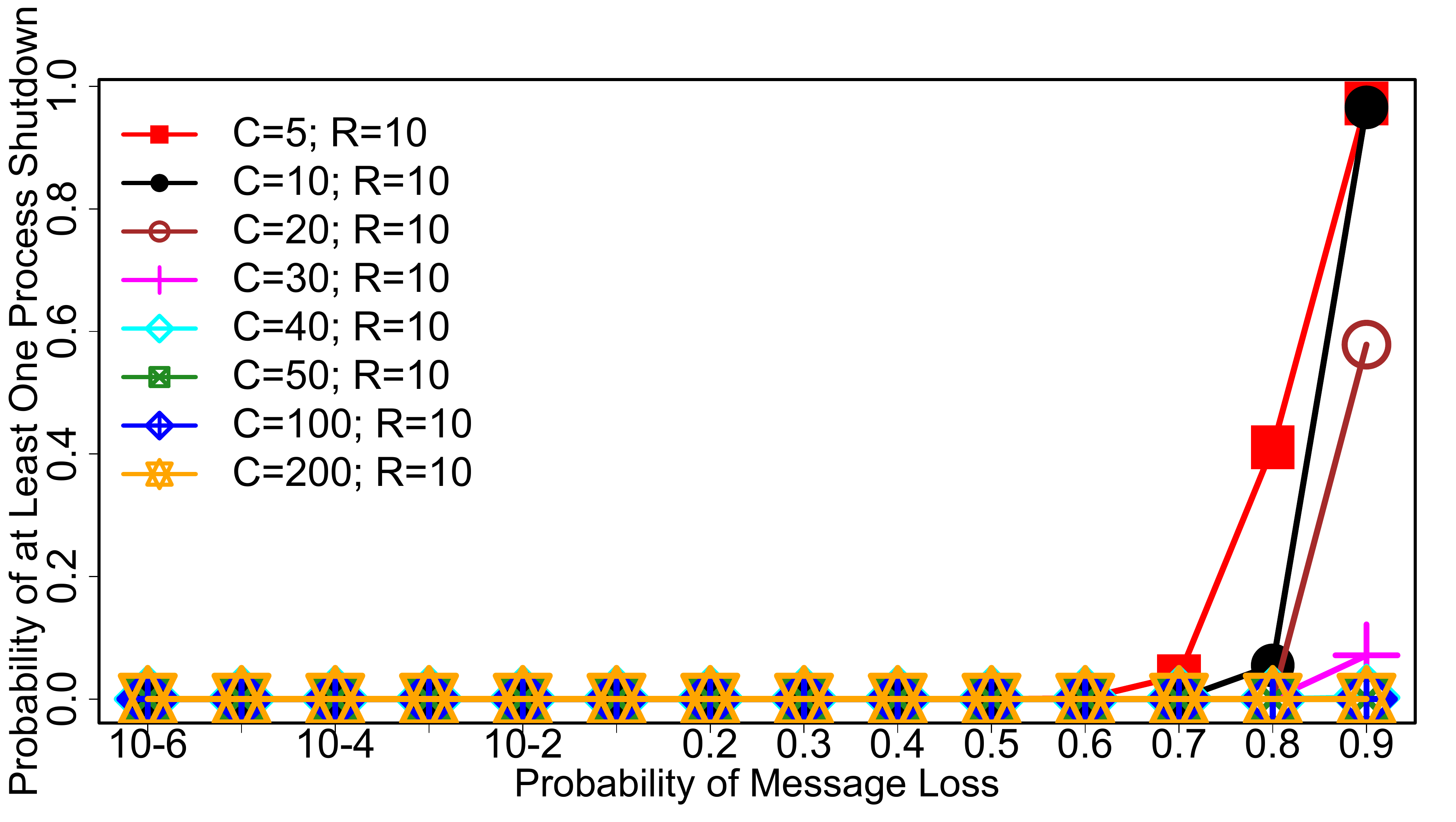}
		\begin{center}
			\vspace*{-10pt}
			\caption{The probability of a correct process crashing itself in a system after $10$ communication rounds ($\mathbb{R}=10$).}\label{pshut10}
			\vspace*{-15pt}
		\end{center} 
	\end{center}
\end{figure}
\begin{figure}[t] 
	\begin{center}
		\includegraphics[scale=0.2]{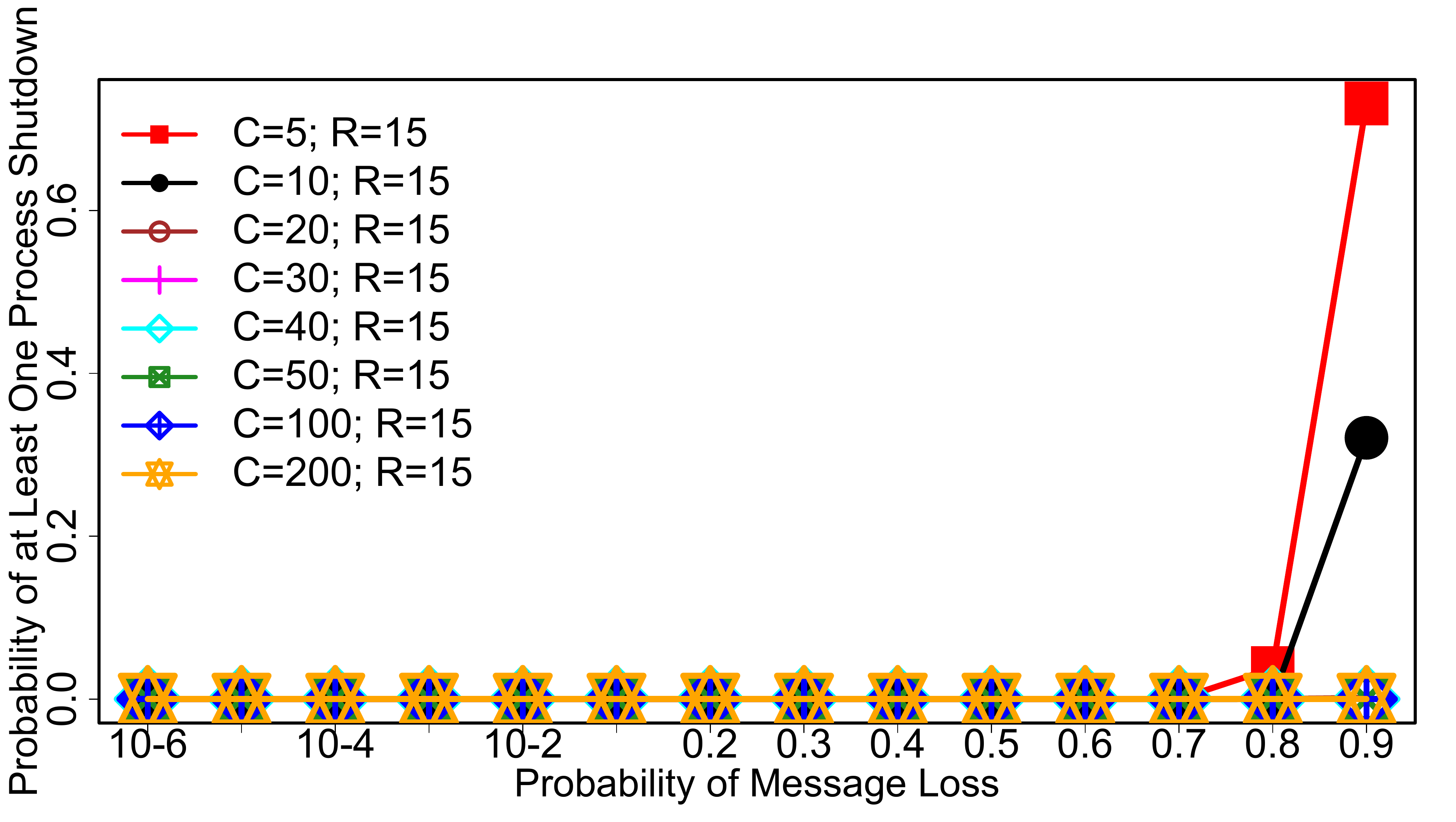}
		\begin{center}
			\vspace*{-10pt}
			\caption{The probability of a correct process crashing itself in a system after $15$ communication rounds ($\mathbb{R}=15$).}\label{pshut15}
			\vspace*{-15pt}
		\end{center} 
	\end{center}
\end{figure}
Our results in Figure~\ref{pshut5}, Figure~\ref{pshut6}, Figure~\ref{pshut10}, and Figure~\ref{pshut15} show that our \textit{RT-ByzCast} algorithm achieves a higher reliability, i.e., a smaller probability of having a correct process crash itself, as the number of non-Byzantine processes, $|\mathcal{C}|$, increases. This aspect of our algorithm is fundamental to the typical applications of the target CPS systems we envisage, e.g., the smart grid, where computing sensing and communication devices are deployed on a large scale.

\textcolor{red}{Moreover, from a different perspective, namely by fixing $|C|$ and varying the window size ($\mathbb{R}$), we show in Appendix~D that increasing the number of communication rounds (i.e., the value of $\mathbb{R}$) increases the reliability of our \textit{RT-ByzCast} algorithm for any number of correct processes.} In fact, with $\mathbb{R}=10$ the probability of a process crashing itself becomes negligible even with up to $60\%$ losses/omissions rate.

\vspace*{-10pt}\subsection{Tolerating A Correct Process Self-Crash}\label{system shutdown}
Despite the small probability of having a process crash itself, we provide in this section a means to allow the system using \textit{RT-ByzCast} to tolerate processes self-crashing while still being available. Let us respectively denote by $Prob(p_{shutdown})$ and $Prob(sys_{shutdown})$ the probability that a process kills itself and the probability that the whole system shuts down. A system where $n=3f+1$ shuts down, in the worst-case, as soon as a single process crashes itself. Hence, an upper bound on the probability that the whole system shuts down can be formally expressed by:
\begin{equation}\label{leastr}
\resizebox{0.94\hsize}{!}{$
	\begin{split}
	Prob(&sys_{shutdown})= 1- (1-Prob(p_{shutdown}))^{2f+1}.
	\end{split}$}
\end{equation}

A system using our \textit{RT-ByzCast} algorithm can be made more reliable while tolerating the maximum number of Byzantine processes. Precisely, we can allow the system to tolerate a single correct process crashing itself, after which any additional correct processes that crashes itself would trigger the whole system to shutdown. To do so, for any desired value of $f$, the maximum number of tolerable Byzantine processes, we choose the total number of processes to be $n =3f+3.$ Note that selecting $n$ as such means that the Byzantine quorum now becomes of size $2f+2$ rather than $2f+1$. Correct processes can gather such a quorum even if one correct process crashes itself and all $f$ processes are~Byzantine. 

With this optimization, the probability that whole system shuts down takes a new value, which is expressed by:
\begin{equation}\label{modr}
\resizebox{0.96\hsize}{!}{$
	\begin{split}
	Prob(&sys_{shutdown})=1-  \\&\sum_{k=0}^{1}\dbinom{n-f}{k} (Prob(p_{shutdown}))^{k} (1-Prob(p_{shutdown}))^{n-f-k}.
	\end{split}$}
\end{equation}

From Figure~\ref{pshut10}, a correct process crashes itself  with a probability~$<~10^6$ in systems using more than 5 correct processes, $\mathbb{R} = 10$ rounds, and experiencing  network omission rate of up to $60\%$. Accordingly, we perform a numerical analysis of the probability that the whole system shuts down as expressed in \eqref{leastr}, and \eqref{modr} and we show our results in Figure~\ref{shutdown}. 

The graphs of Figure~\ref{shutdown}  highlight the increase in reliability between the two system design choices showing that a system implementing \textit{RT-ByzCast} operates with a very high reliability (negligible probability of system shutdown).

Enhancing the reliability of the system further, by allowing it to tolerate more correct processes crashing themselves, is possible. However, doing so requires the system to have a mechanism of globally detecting aborted processes that crash themselves. We further elaborate on that in what follows. 

\begin{figure}[t] 
	\begin{center}
		\hspace*{-8pt}
		\includegraphics[scale=0.22]{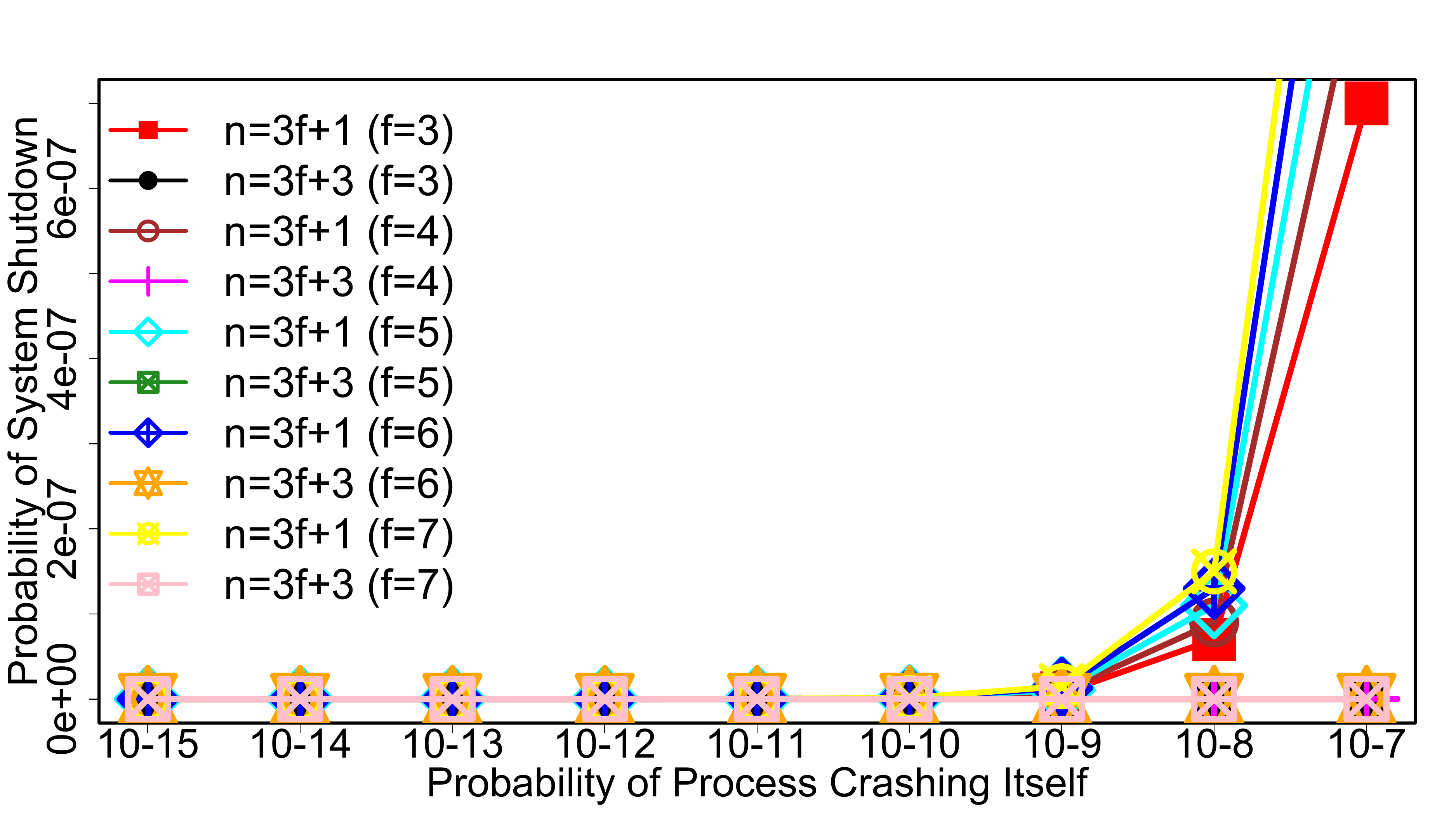}
		\begin{center}
			\vspace*{-10pt}
			\caption{The probability of the whole system shutting down, when $\mathbb{R}=10$, w.r.t. different values of the probability of a correct process crashing itself under $60\%$ network loss rates (taken from Figure~\ref{pshut10}).}\label{shutdown}
			\vspace*{-15pt}
		\end{center} 
	\end{center}
\end{figure}

\subsection{Tolerating Any Number of Processes Self-Crashing}\label{sec: multiple crashes}
In this section, we introduce a way that allows us to design, using \textit{RT-ByzCast}, systems that can tolerate any number of correct processes crashing themselves. The solution is based on two main things: (i) over-provisioning and (ii) self-crash detection. To over-provision in system design, namely means to include additional redundant process replicas. In other words, to tolerate $f$ Byzantine processes and $rep$ simultaneous process self-crashes in the system ($\forall f,rep$) we choose the total number of processes, $n$, to satisfy  $n~=~3f~+~rep+1.$

Processes then will execute our \textit{RT-ByzCast} using Byzantine quorums depending on the value of $n$, which can only monotonically decrease (due to processes crashing themselves). Designing a system as such means that the system can tolerate up to $rep$ correct processes crashing themselves simultaneously (for any value of $rep$). 

However, in order to support system liveness for up to $rep$ correct processes crashing themselves, the system would need to detect, within $\mathbb{R}$ rounds after having crashed, every process that crashed itself. This detection allows other correct processes to update $n$ and the needed quorums, without resulting in a cascade of self crashes leading to system shutdown. To this end, we describe in what follows a mechanism to detect crashed processes within $\mathbb{R}$ rounds after~crashing.

\subsubsection*{Detecting a Correct Process in the \textit{Dead-State}}\label{detect dead-state}

We recall that when a process $p_i$ executes the proof-of-life function as well as when that process issues a broadcast, $p_i$ needs to issue a message that should be seen by $2f+1$ processes within a window of $\mathbb{R}$ rounds; otherwise $p_i$ would crash itself (case 4 of Section~\ref{RT-ByzCast}). In this sense, having any set of $2f+1$ processes not hearing from $p_i$ in a window of $\mathbb{R}$ rounds indicates that $p_i$ has crashed itself (or is Byzantine). Note that a process $p_j$ can hear from $p_i$ either directly or via other processes. 

The algorithm for detecting process self-crashes (and Byzantine processes mimicking a self-crash behavior) works as follows. Every process now includes in its message the round number during which it has sent the message. Every process $p_i$ appends to every message it sends a list $\mathcal{L}_{p_i}$ of size $n$, relative to the $n$ processes in the system. $\mathcal{L}_{p_i}$ contains, for every process $p_j \in \mathit{\Pi}$, the time-stamp relative to last time that $p_j$ sent a message (based on $p_j$'s messages that $p_i$ received). Note that a stored time-stamp relative to process $p_j$ should have an associated proof of validity, that being the signature of $p_j$ on a message holding that time-stamp (invalid time-stamps are ignored). Every process merges all the $\mathcal{L}$ lists received such that the time-stamp with largest valid value for each process is kept. In every round $r:r>\mathbb{R}$, if a process $p_i$ receives $2f+1$ lists such that the time-stamp relative to some process $p_j$ is less than the current round minus $\mathbb{R}$ in all received $\mathcal{L}$, then $p_i$ does the following:
(1) claims $p_j$ as crashed, (2) updates $n$ to $n-1$ and updates as well the new corresponding Byzantine quorum size, and (3) ignores all messages relative to $p_j$ that may arrive at a later time. 





\vspace*{-10pt}\subsection{\textcolor{red}{Performance of \textit{RT-ByzCast}}}\label{key}

\textcolor{red}{In addition to evaluating the reliability of \textit{RT-ByzCast}, we evaluate how fast it can deliver messages. In Appendix~E, we also show the incurred cost that \textit{RT-ByzCast} has on the network bandwidth. \textit{RT-ByzCast} has a latency of at most $3\mathbb{R}$ to deliver a message to all correct processes. Recall that we define $\mathbb{R}$ as a slack period of collective re-transmissions of some message $m$ sent by process $p$. The value of $\mathbb{R}$ should be chosen such that $p$ can know with high probability that at least $2f+1$ processes received $m$.}  


\textcolor{red}{We now simulate the value of $\mathbb{R}$ (in communication rounds) as follows. We initially set $\mathbb{R}=1$ round and run our reliability experiments for $2\times 10^{5}$ times. We gradually increase the value of $\mathbb{R}$ until no node aborts in all $2\times 10^{5}$ repetitions. In our experiments, we assumed the worst case situation where all ``$f$" nodes are Byzantine and can obstruct our algorithm by not forwarding messages' signatures. Figure~\ref{fig:rounds} represents the value  of $\mathbb{R}$ for 10, 50, 100 and 200 nodes, depending on the probability of message loss.  First, as one may expect the empirical value of $\mathbb{R}$ increases with the probability of message loss. Second, we observe that increasing the system's size decreases the value of $\mathbb{R}$. This may appear unexpected, however, it is due to the fact that nodes forward messages in an all-to-all manner where messages can be probabilistically lost. Therefore, message losses have a lower impact on the broadcast delay in large scale systems. For example, with 90\% message losses, $\mathbb{R}$ is 41 (7) rounds with 10 (200) nodes resp.} 

\begin{figure}[t] 
	\begin{center}
		\includegraphics[width=0.85\columnwidth]{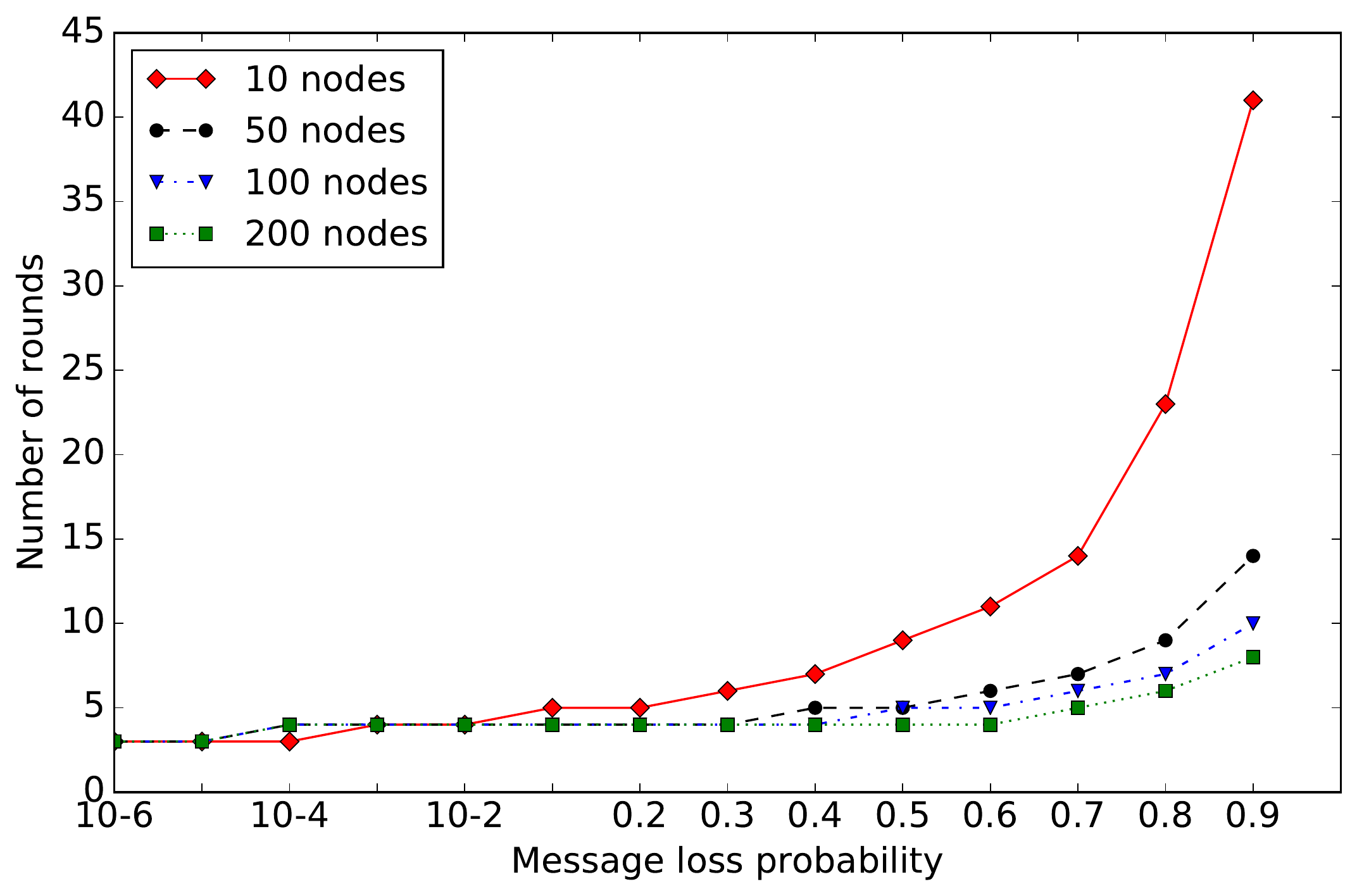}
		\begin{center}
			\vspace*{-10pt}
			\caption{\textcolor{red}{The size of the time window ($\mathbb{R}$) depending on the system size and the message loss probability.}}\label{fig:rounds}
			\vspace*{-15pt}
		\end{center} 
	\end{center}
\end{figure}

\textcolor{red}{To obtain the latency of \textit{RT-ByzCast}, we empirically evaluate the worst case computational delay ($d_{max}$) needed per node per round. This delay includes the time a node needs to process all received messages and prepare the corresponding messages to be sent.
Nodes use the ecdsa-256 encryption scheme~\cite{johnson2001elliptic}, which is known for being fast and for having relatively short signatures. The overall worst case latency of \textit{RT-ByzCast} would then be computed as $3\times\mathbb{R}\times d_{max}$. We illustrate this overall latency in Figure~\ref{fig:broadcasttime}. We observed from our experiments that having more Byzantine nodes in the system puts a higher computing load on correct nodes.  The variations in Figure~\ref{fig:broadcasttime} for a given system size result from two main factors: (i) the value of $\mathbb{R}$ increases with message losses (Figure~\ref{fig:rounds}), and (ii) the value of $d_{max}$ decreases with higher message losses, depending on the maximum number of simultaneously received messages per round. With low loss rates (from $10^{-6}$ to $10^{-4}$ in Figure~\ref{fig:broadcasttime}), a slight increase in the value of $\mathbb{R}$ increases the broadcast time noticeably. As message losses become more frequent (e.g., from $10^{-4}$ to $0.5$) the value of $d_{max}$ decreases, reducing the broadcast delay, until the value of $\mathbb{R}$ increases sufficiently to compensate and increase again the broadcast time (e.g., from $0.5$ to $0.9$).}


\textcolor{red}{Our numbers in Figure~\ref{fig:broadcasttime} show that without any performance optimizations, and despite message losses, \textit{RT-ByzCast} can meet the timing constraints of various applications, such as (i) power system automation and substation automation applications (IEC 61850-5 standard~\cite{IEC61850}), e.g., transfer of automation functions (TT3 class messages) $\approx 100$~ms (systems size $\leq 50$), slow speed auto-control functions, time-tagged alarms, event records,  set-values read/write operations (TT4 class messages) $\leq 500$~ms, (ii) continuous control (e.g. temperature-driven) applications $\leq 1$~s~\cite{Dav},  and (iii) operator commands of SCADA applications ($\leq 2$ s)~\cite{Moga:2016}.}

\textcolor{red}{However, we believe that 
\textit{RT-ByzCast}'s latency has a substantial margin for
improvement, whose avenues we discuss in the conclusions.}



\begin{figure}[t] 
	\begin{center}
		\includegraphics[width=0.85\columnwidth]{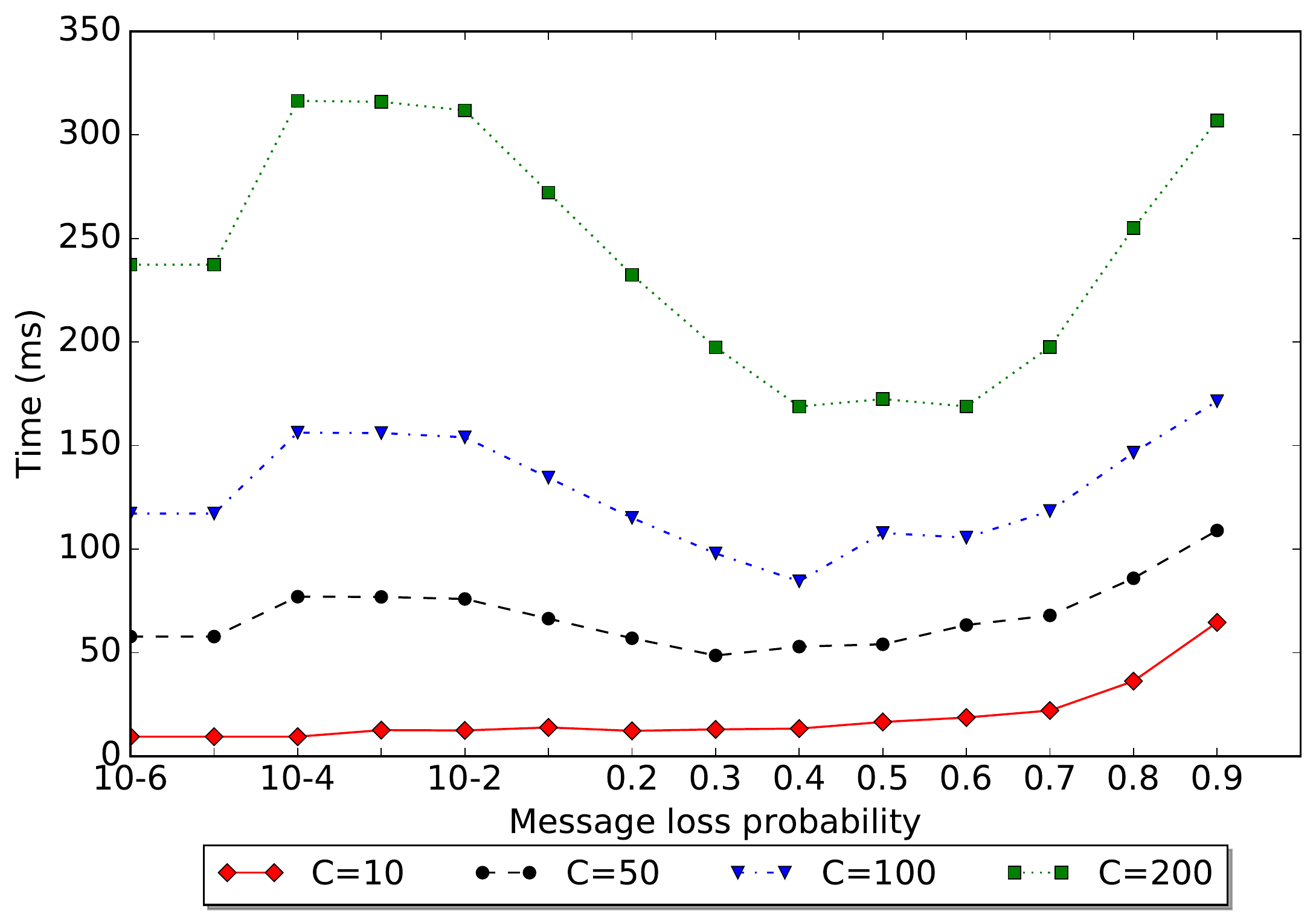}
		\begin{center}
			\vspace*{-10pt}
			\caption{\textcolor{red}{Total broadcast time depending on the system size and the message loss probability.}}\label{fig:broadcasttime}
			\vspace*{-15pt}
		\end{center} 
	\end{center}
\end{figure}

\vspace*{-10pt}\section{Reviving Processes in The \textit{Dead-State}}\label{Dead-State Revival}

In order to allow processes to leave the \textit{Dead-State}, we propose a modification to our \textit{RT-ByzCast} algorithm. A process $p_i$ that issues a broadcast at some round $r$ sends now \texttt{RTBRB-broadcast($p_i,r,v;\Phi_{p_i}$)} instead of \texttt{RTBRB-broadcast($p_i,v;\Phi_{p_i}$)}. In other words, $p_i$ includes the round in which the broadcast is issued (first sent) in the payload of the \texttt{RTBRB-broadcast()} message that $p_i$ sends. The behavior of our \textit{RT-ByzCast} algorithm of Section~\ref{RT-ByzCast} remains the same.

A process in the \textit{Dead-State} mimics the behavior of a crashed process in the \textit{fail-stop} model~\cite{Schneider:1984:BGA:190.357399}. However, unlike an actual crashed process, a process $p_i$ in the \textit{Dead-State} continues to listen and process received messages. In other words, a process $p_j$ that is in the \textit{Dead-State} executes our \textit{RT-ByzCast} algorithm, however with the exception that $p_j$ refrains from sending any message. If process $p_j$ (in the \textit{Dead-State}) detects that some value $v$ relative to process $p_i$ has been signed by at least $2f+1$ processes (i.e., $|\mathcal{R}_{echo}(p_i,r,v)|> 2f$), then $p_j$ performs a check: if the current round, $r_{cur}$, satisfies $r_{cur}-2\mathbb{R}>r$, where $r$ is the round in which $v$ was issued, then $p_j$ delivers the message $v$ (if it has not done so), shifts to the \textit{Alive-State}. Then $p_j$ transmits during every round in $[r_{cur}+1,r_{cur}+1+\mathbb{R}]$ the message \texttt{Deliver($(p_i, r,v,\textcolor{red}{sigs});*$)} and resumes as in Section~\ref{RT-ByzCast}.


\vspace*{-10pt}\section{The Case of a Dynamic System}\label{dynamic system}
So far in the paper, we have discussed systems where the maximum number of processes $n$ is static and known to all processes. In this section, we consider a dynamic system where the maximum number of processes in the system is unknown, as processes can randomly join and leave. For consistency, as the system now allows processes to join at any point in time, we assume the following: a process that transitions to the \textit{Alive-State} at round $r$ is not required to deliver any messages that were broadcast at rounds $<r$.

We consider that correct processes now can exist, at any point in time, in one of three possible states: \textit{Alive-State}, \textit{Dead-State}, and \textit{Pending-State}. Only those processes that exist in the \textit{Alive-State} are the ones that should guarantee the RTBRB properties. We assume that initially $x>3$ processes exist and are awakened by the clock initialization signal. We assume that these processes constitute a \textit{trusted pool} and that their identity is known to any process that wishes to join the system later. Some processes within the trusted pool, at most $\lfloor\frac{x-1}{3}\rfloor$, can be Byzantine. Our trusted pool can be viewed as a trusted authority, which is a widely adopted assumption in Byzantine group communication protocols~\cite{Miguel-correia2002efficient-byzantine-resilient-116,reiter}. In fact, in the absence of a centralized
admission  control  scheme, a  trusted entity  becomes a necessity for allowing joins in an intrusion-tolerant system~\cite{Danezis:2005:SDR:2156732.2156750}.  Processes in the trusted pool constitute the initial system and execute our \textit{RT-ByzCast} algorithm. 

Hence, the total number of processes, $n$, is initially equal to $x$ ($n$ however is not fixed). $n$ is the number of processes that exist in either the \textit{Alive-State} or the \textit{Dead-State}. We consider that at any point in time at most $f=\lfloor\frac{n-1}{3}\rfloor$ can be Byzantine. 

We assume that every process maintains a list $\mathcal{ID}$ that contains the ids of the processes in the system. Initially the $\mathcal{ID}$ list contains the ids of all processes in the trusted pool. A process $p_j$ that wants to join the system at round $r$, initializes its state to the \textit{Pending-State}, sets $n=x$, and defines $\mathcal{ID}$ as the list of ids of all processes in \text{trusted pool}. 

$p_j$ begins to send at the beginning of every round in $[r+1,r+1+\mathbb{R}]$ the heartbeat message, \texttt{HB($\{p_j,key_j,r\};\Phi_{p_j}$)}, to the processes in the trusted pool. Upon the receipt of a \texttt{HB($\{p_j,key_j,r\};\Phi_{p_j}$)} in round $r':r'\in[r+1,r+1+\mathbb{R}]$ a process $p_i$ in the trusted pool executes the following, if $p_i$ is not sending some message for another potential joining process. $p_i$ verifies the received heartbeat. Then $p_i$ creates a \texttt{Join($\{p_j,key_j,r\}\{n,\mathcal{ID}\};\Phi_{p_j},\Phi_{p_i}$)} message and sends this message in rounds $r'':r''\in[r'+1,r+1+\mathbb{R}]$ it to all processes in the trusted pool, and to $p_j$ as well.

If a process $p_k$ in the trusted pool receives \texttt{Join($\{p_j,key_j,r\}\{n,\mathcal{ID}\};\Phi_{p_j},...,\Phi_{p_i}$)} at some round $r':r'\in[r+1,r+1+\mathbb{R}]$, $p_k$ verifies the message received and validates that the values of $n$ and $\mathcal{ID}$ match with its local values. Then, considering that $\mathcal{SIG}(p_k)$, contains the signatures that $p_k$ received so far relative to $\{p_j,key_j,r\}\{n,\mathcal{ID}\}$ ($p_k$'s signature is included in $\mathcal{SIG}(p_k)$),  process $p_k$ executes $\mathcal{SIG}(p_k) = \mathcal{SIG}(p_k)\bigcup\{\Phi_{p_j},...,\Phi_{p_i}\}$. If $\mathcal{SIG}(p_k)$ contains signatures from at least $\lfloor\frac{x-1}{3}\rfloor +1$ processes, then $p_k$ sends in rounds $r'':r''\in[r'+1,r+1+\mathbb{R}]$ \texttt{Join($\{p_j,key_j,r\}\{n,\mathcal{ID}\};\mathcal{SIG}(p_k)$)} to all processes in the trusted group and to $p_j$ as well. 
	

The joining process $p_j$ also executes the same steps when receiving  \texttt{Join($\{p_j,key_j,r\}\{n,\mathcal{ID}\};\Phi_{p_j},...,\Phi_{p_i}$)} messages. If at some round $r'\leq r+1+\mathbb{R}$, $\mathcal{SIG}(p_j)$ relative to some $(n,\mathcal{ID})$ contains signatures from at least $2\lfloor\frac{x-1}{3}\rfloor +1$ processes in the trusted pool, $p_j$ transitions to the \textit{Alive-State}, updates its parameters accordingly (to be equivalent to the new values of $n+1$ and $\mathcal{ID}\bigcup p_j$), and invokes an \texttt{RTBRB-broadcast(join$(p_j,key_j)$;$\Phi_{p_j}$)}. Otherwise, i.e., if $p_j$ fails to collect enough signatures on some values of $n$ and $\mathcal{ID}$ by round $r+1+\mathbb{R}$, then $p_j$ remains in the \textit{Pending-State}, quits sending any messages, and tries to join the system again after a random duration~elapses.

All processes that take part of the \texttt{RTBRB-broadcast(join$(p_j,key_j)$;$\Phi_{p_j}$)} do not count $p_j$'s signature as part of any quorum (but simply use it for verifying the validity of the message sent). Any process that \texttt{RTBRB-delivers(join$(p_j,key_j)$;*)}, except $p_j$, updates the set of public keys to include $key_j$ and the local parameters to $n=n+1$ and  $f=\lfloor\frac{n-1}{3}\rfloor$, $\mathcal{ID}\bigcup p_j$. That concludes the joining procedure for processes. 
A process $p_l$ that wants to leave the system, issues an \texttt{RTBRB-broadcast(leave$(p_l,key_l)$;$\Phi_{p_l}$)}. $p_l$ remains in the \textit{Alive-State} until $p_l$ delivers the message relative to its leave request broadcast, after which $p_l$ shifts to the \textit{Dead-State}. If it wishes, $p_l$ can join the system later.

\vspace*{-10pt}\section{RT-ByzCast Potential Application Domains}\label{RT-ByzCast applicability}
The problem we address --- disseminating information in a real-time Byzantine-resilient manner given communication synchrony disruptions --- is of fundamental importance to any critical distributed system seeking dependability and security in the CPS context. RTBRB is important to, e.g., (1)~achieve a real-time common view of the state of the system by independent processes, (2)~distribute sensory information in a coherent way or (3) help ensure the consistency of replicas driving a same end-component, e.g., an actuator. 

The difficulty of the problem lies in ensuring predictability and resilience despite  communication uncertainties, faults, and attacks that hinder the synchronism needed to meet real-time deadlines. In this sense, \textit{RT-ByzCast} showcases an algorithmic design solution to circumvent the impossibilities that prevent traditional modular distributed computing approaches from achieving real-time Byzantine-resilient information dissemination.  As a well-contained primitive, \textit{RT-ByzCast} can be integrated in implementations of existing systems and frameworks, for example those related to distributed control systems (DCSs)~\cite{Galloway} which include monitoring and control applications for factory automation, substation automation and smart grids. In fact, real-time execution platforms for DCSs, such as FASA~\cite{FASA,FTFT}, have architectures that adhere to the assumptions considered in this paper; thus \textit{RT-ByzCast} can be of substantial use to both the underlying system-monitoring algorithms as well as to the control applications being run (e.g., power control applications~\cite{Dav}). Other application systems include ship-board DCSs requiring distributed real-time data services~\cite{qosmanagement,jitrtr}, traffic control and agile manufacturing that require ``fresh" data reflecting real-world status~\cite{qosmanagement,ptudpr,rtpb,jitrtr}, and multi-gateway \emph{SmartData} construct protocols for CPSs on Wireless Sensor Networks or on the Internet of Things~\cite{guto}.

\vspace*{-5pt}\section{Related Work}\label{related work}

\textcolor{red}{The work in this paper has evolved from both~\cite{flaviu} and~\cite{AMP}, in what concerns the timing/synchrony aspect. In~\cite{flaviu}, it was assumed that all non-faulty processes remain connected synchronously, regardless of any process and network failures. This strong assumption about the network was too ideal, in terms of scale and timing behaviour. This resulted in a poor performance in practice (latency of $\approx 20$~s) that limit~\cite{flaviu}'s application areas. The limited performance is mainly attributed to deterministic network model (which in turn affected the solution design). Moreover, the system model in~\cite{flaviu} did not allow processes that malfunction (violates assumptions) to know that they are being treated as faulty by the model. 
 In contrast, RT-ByzCast provides latencies in the range of milliseconds (from few to few hundreds depending on system size and message loss). We achieve this performance gain by forcing processes to operate within the needed delays: we exclude processes that are incapable of meeting the desired timing requirements. Moreover, in our model, processes that violate timeliness assumptions transition to the \textit{Dead State} and hence are aware of their ``non-correctness" (in terms of time).}


\textcolor{red}{In~\cite{AMP}, the timeliness problem was addressed by what the authors called weak-fail-silence: despite the capability of the transmission medium to deliver messages reliably and in real-time, the protocol should not be agnostic of potential timing or omission faults (even if sporadic). This notion was embedded in a bounded omission specification that weakens the basic fail-silence (crash) hypothesis. In our paper, we make a significant advance, by taking a further step: providing these reliable real-time communication guarantees, in environments with much higher uncertainty levels (faults and attacks). In fact, the bounded omissions assumption of~\cite{AMP} could not be taken as is, if we were to tolerate such higher and more uncertain fault sets (as we consider): it could easily lead to system unavailability in faulty~periods.}

\textcolor{red}{The rest of the literature on broadcast primitives, to the best of our knowledge, either does not take into account timeliness and maliciousness or addresses them separately. We next summarize few such related works.}

\vspace*{-10pt}
\subsection{Byzantine Reliable Communication}

In order to withstand unpredictable security threats and software unreliability arising within networks, various works have investigated how to provide a reliable broadcast capable of tolerating arbitrary process behavior. A celebrated algorithm in distributed computing is that of Bracha and Toueg~\cite{Bracha:1987}, which implements reliable broadcast in an asynchronous system of $n$ processes with at most $f<\frac{n}{3}$ Byzantine processes (proved to be an upper bound on the number of tolerable Byzantine processes~\cite{Bracha:1985,Bracha:1987}). The broadcast algorithm of~\cite{Bracha:1987}, similar to the one proposed in this paper, relies on echoing enough messages before delivering a value. Later works tried to further optimize this algorithm, e.g., in terms of the total number of asynchronous communication rounds needed for termination~\cite{raynal2015}. On a different level Obenshain et al.~\cite{yair} design and construct an intrusion-tolerant overlay capable of tolerating Byzantine actions, based on the key understanding that no overlay node should be trusted or given preference. The authors use a maximal topology with minimal weights 
to prevent routing attacks at the overlay level and rely on source routing augmented with
redundant dissemination methods to limit the effect of compromised forwarder processes.%

Unlike our \textit{RT-ByzCast} algorithm, the solutions of~\cite{Bracha:1985,Bracha:1987,raynal2015,yair} do not provide timeliness guarantees on message delivery, even when applied to weakly synchronous networks such as the one we consider in this paper.
\vspace*{-10pt}
\subsection{Timeliness Communication Guarantees}

Given the need for predictable responsiveness in various applications today, many efforts have been made to devise timely broadcast (or routing/communication) algorithms, i.e., communication protocols with known fixed delays.

Timeliness on web communication, for example, has been addressed in content distribution networks~\cite{1019427,Pallis:2006:IPC:1107458.1107462} such as Akamai~\cite{Dilley:2002:GDC:613356.613741}. The concept behind such distribution networks is to place data in a close geographical proximity from its consumers and hence minimize latency. In a slightly different context, timeliness has been also addressed at the level of data center networks~\cite{Vamanan:2012:DDT:2342356.2342388,Wilson:2011:BNL:2018436.2018443}. The proposed solutions in this area mainly advocate modifications of network devices. Bessani et al. propose Jiter~\cite{jiter}, an application-layer routing, as a means to provide message latency and reliability assurances for control traffic in wide-area IP networks. Jiter routes deadline constrained control messages using an overlay network created on top of  multi-homed communication infrastructure. Babay et al.~\cite{yair1} present an overlay transport service that can provide highly reliable communication while meeting stringent timeliness guarantees. Their scheme relies on an analysis of real-world network data, upon which they develop timely dissemination graphs and specify targeted redundant transmissions to timeliness and reliability. 

Existing efforts have also investigated providing end-to-end guarantees in networks. Jacob et al.~\cite{endtoendreal-time} propose an  approach  to  integrate a  wireless  real-time  communication  protocol  into  CPS. Their approach decouples communication from application tasks. They devise a protocol based on dynamically establishing contracts between source/destination devices and the networking protocol. Another set of protocols known as resource  reservation  protocols (RSVP) combine flow specification, resource reservation, admission  control,  and  packet  scheduling  to  achieve  end-to-end QoS~\cite{rsvp}. Relying on software-defined networks, Kumar et al.~\cite{SDN} propose a framework that synthesizes network paths which can meet the requisite delay requirements for real-time flows.

On a different level Guerraoui at al.~\cite{Dav} devise an algorithm that allows processes to communicate in timely fashion via a real-time distributed shared memory (DSM). 
Using their DSM, write operations of a correct process become consistently visible to all alive processes within a known fixed delay. Their approach relies on bundling data messages with control traffic, namely traffic relative to the failure detector component used for monitoring processes in distributed control systems. 

However, all above solutions\cite{1019427,Pallis:2006:IPC:1107458.1107462,Dilley:2002:GDC:613356.613741,Vamanan:2012:DDT:2342356.2342388,Wilson:2011:BNL:2018436.2018443,jiter,yair1,endtoendreal-time,rsvp,SDN,Dav}, unlike \textit{RT-ByzCast}, cannot handle malicious process behavior and hence would fail if any process is compromised.




\vspace*{-10pt}
\section{Conclusion}\label{conclusion}
This paper studied how to realize the RTBRB abstraction: a
real-time Byzantine-resilient reliable broadcast in the presence of unbounded
communication losses and delays. We first showed that implementing RTBRB is challenging and in fact is impossible
under traditional paradigms that consider failure detection mechanisms
working independently of the distributed algorithms.  To circumvent it, we proposed
\textit{RT-ByzCast}, an algorithm that deploys temporal and spatial diffusion of messages, and signature aggregation in a sliding time-window to mask losses and
track process connectivity. \textit{RT-ByzCast} bounds network timing uncertainties by reconciling two mechanisms: cooperative round-based message retransmission
involving all processes; and proactive self-crashing of misbehaving
processes.  


We proved that \textit{RT-ByzCast} indeed implements the desired RTBRB
abstraction and we showed that it does so robustly and efficiently: we
evaluated the reliability of our algorithm, showing that it can
tolerate quite high loss rates ($\approx 60\%$) whilst still
delivering the real-time service and ensuring a negligible probability
that any correct process crashes itself, hence guaranteeing system
survivability. \textcolor{red}{We also showed that \textit{RT-ByzCast} meets the
timing constraints expected by a wide range of CPS and IoT applications,
whose latency requirements are normally inversely proportional to the
system size: RTBRB delays are in the order of the dozens of
milliseconds on average, ranging from few milliseconds to hundreds of milliseconds, as system size goes
from 10 to 200 nodes.
Furthermore, \textit{RT-ByzCast} design and minimal environment assumptions,
simplify its integration in existing architectures.} 


\textcolor{red}{By providing the first solution to the RTBRB problem in environments
where synchrony constraints co-exist with maliciousness, the
objective of this paper has been met: discovering and proving
impossibility results countering common beliefs about failure
detection, which will guide future research in these environments, as
well as developing algorithms to circumvent them which, whilst
non-optimised, already provide suitable performance for real world
CPS/IoT applications.}

\textcolor{red}{However, we believe
\textit{RT-ByzCast}'s latency has a substantial margin for
improvement, to meet more demanding timeliness needs, still in large
scale system deployments.  We plan to address these optimisations in
future work, such as using gossip paradigms to decrease the nodes'
transmission costs, and using multisignature/aggregate signature
schemes to amortize the costs of signatures.}



\vspace*{-5pt}\section{Acknowledgments}\label{acks}
This work is in part supported by the Fonds National de la Recherche Luxembourg through PEARL grant
FNR/P14/8149128.
\vspace*{-5pt}
\renewcommand{\baselinestretch}{0.93}
\bibliographystyle{IEEEtran}
\bibliography{RTBCast}
\vspace*{-20pt}
\renewcommand{\baselinestretch}{0.9}
\begin{IEEEbiography}[{\includegraphics[width=1in,height=1.25in,clip,keepaspectratio]{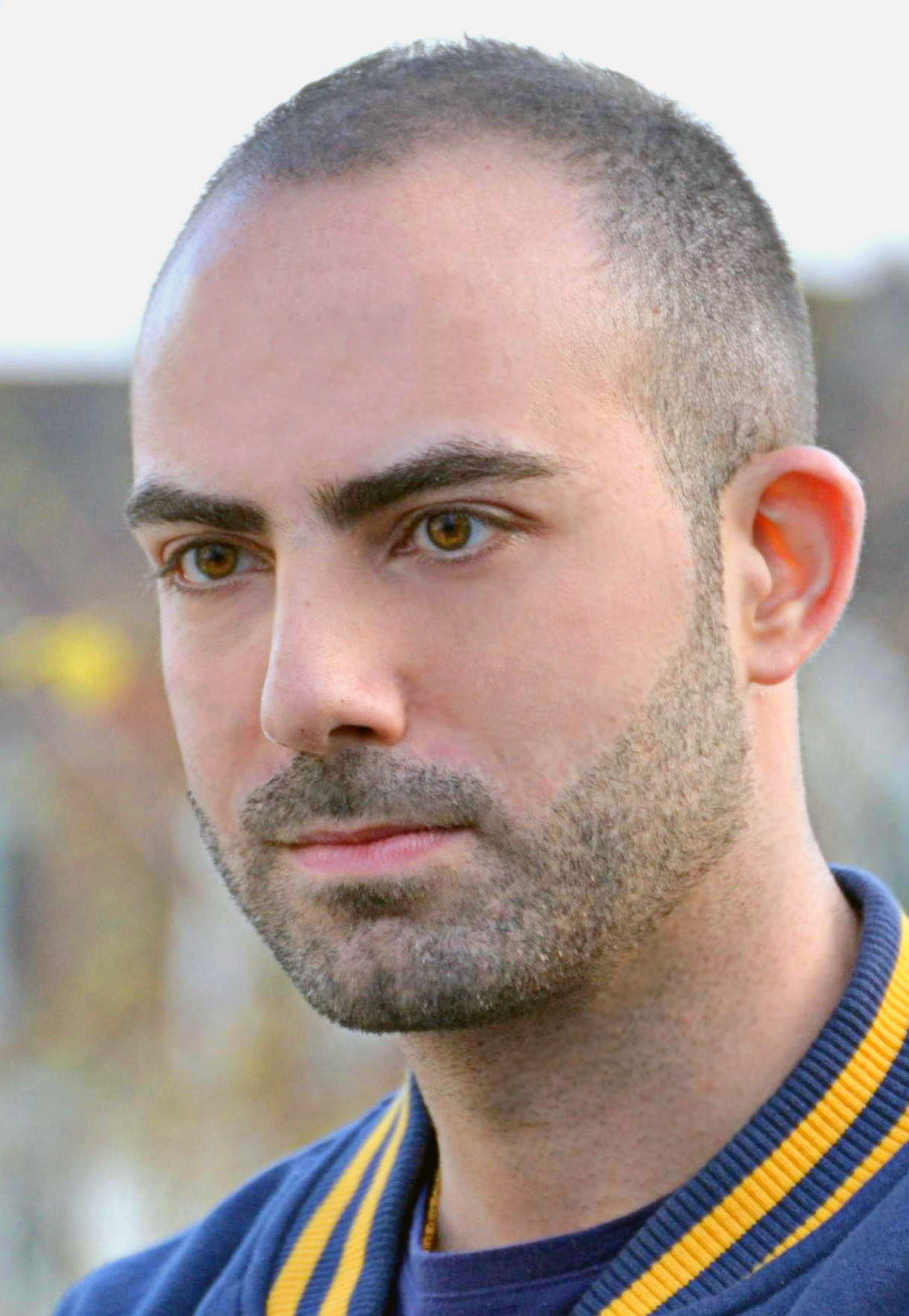}}]{David Kozhaya} 
	is a Scientist at ABB Corporate Research, Switzerland.
He received his PhD degree in Computer Science in 2016, from EPFL, Switzerland, where he was granted a fellowship from doctoral school. His primary research interests include reliable distributed computing, real-time distributed systems, and fault- and intrusion-tolerant distributed algorithms. His past work experiences span across interdisciplinary domains ranging from research, teaching programming languages and computer literacy, financial and market analysis, and management of non-profit organizations. 
\end{IEEEbiography}

\begin{IEEEbiography}[{\includegraphics[width=1in,height=1.25in,clip,keepaspectratio]{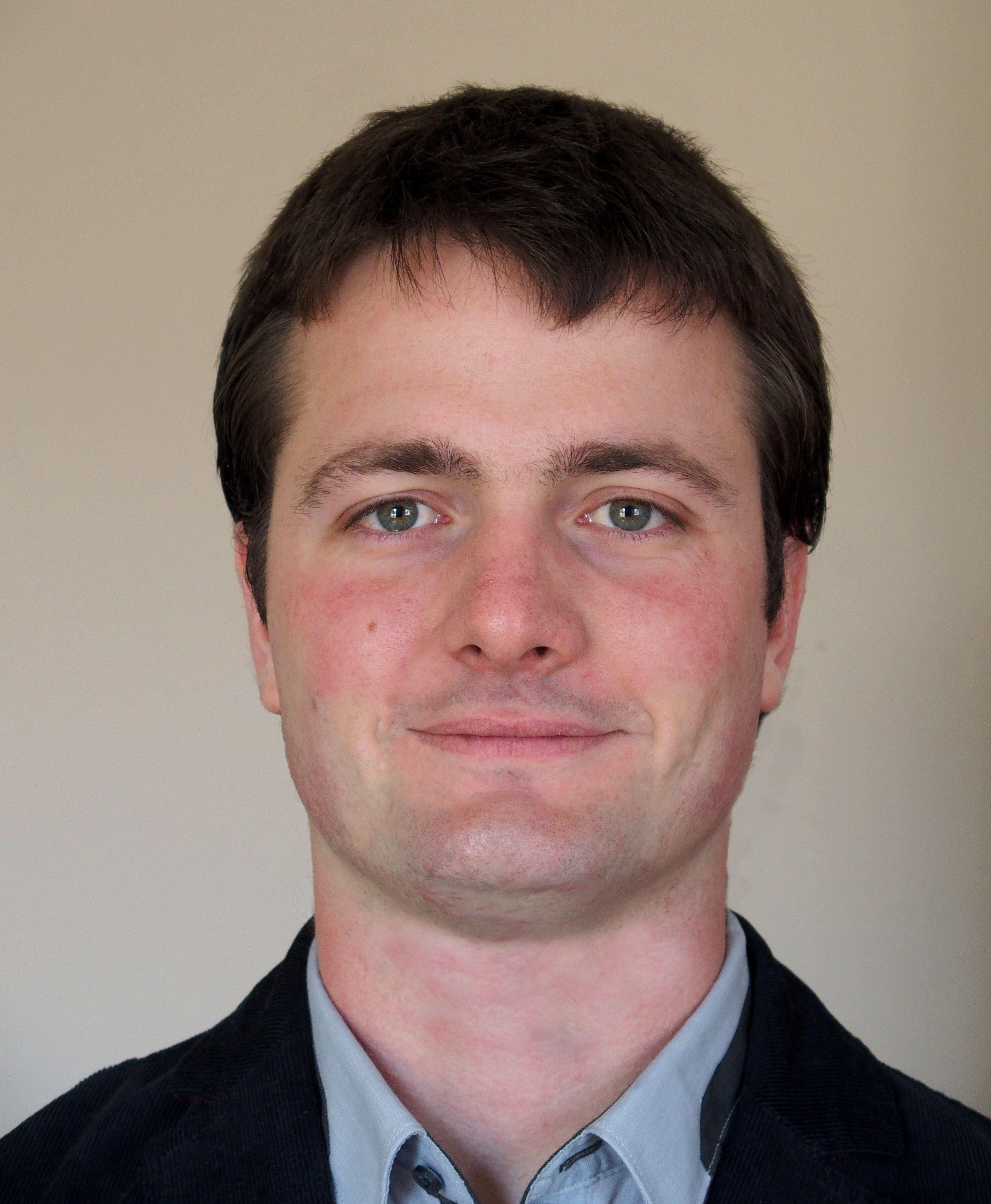}}]{J\'er\'emie Decouchant} \textcolor{red}{is a Research Associate at SnT, University of Luxembourg. He received his Ph.D in Computer Science from the University of Grenoble-Alpes, France. Before that he obtained an engineering degree (MSc) from the Ensimag engineering school, in Grenoble. His research focused on the design and analysis of mechanisms to protect distributed collaborative systems against selfish or Byzantine behaviors. More recently, he has been designing distributed and privacy preserving genomic information processing workflows.}
\end{IEEEbiography}

\begin{IEEEbiography}[{\includegraphics[width=1in,height=1.25in,clip,keepaspectratio]{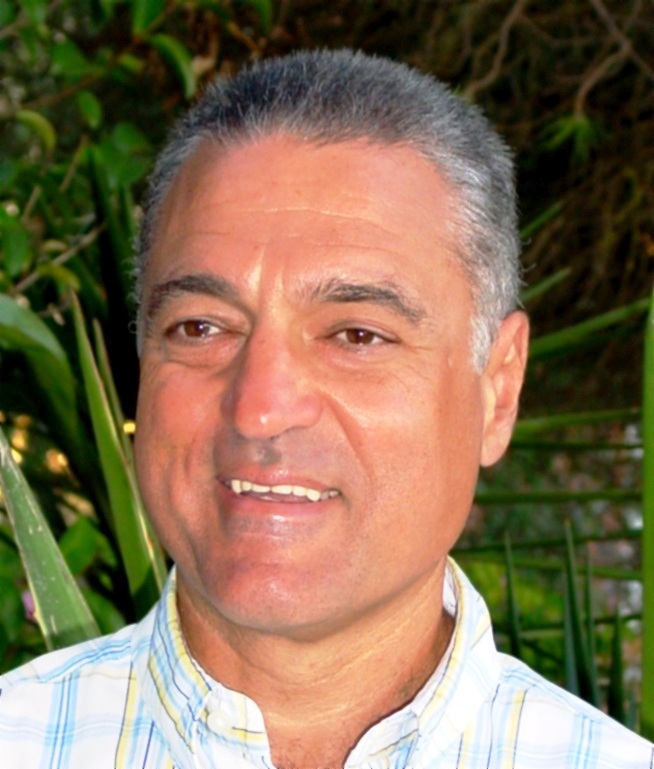}}]{Paulo Esteves-Verissimo} is a Professor and FNR PEARL Chair at the University of Luxembourg (UL), and head of the CritiX research group at UL's SnT Centre (http://wwwen.uni.lu/snt). He is adjunct Professor of the ECE Dept., Carnegie Mellon University. Previously, he has been a Professor of the Univ. of Lisbon. Verissimo is Fellow of IEEE and of ACM. He is Chair of the IFIP WG 10.4 on Dependable Computing and Fault-Tolerance and vice-Chair of the Steering Committee of the DSN conference. He is associate editor of the IEEE Transactions on Computers. He is interested in secure and dependable distributed architectures, middleware and algorithms for: resilience of large-scale systems and critical infrastructures, privacy and integrity of highly sensitive data, and adaptability and safety of real-time networked embedded systems. He is author of over 180 peer-refereed publications and co-author of 5 books.
\end{IEEEbiography}






	\appendices 
	\section{Proof of Theorem~2}\label{proof imp FD}
	\vspace*{-3pt}
	\textcolor{red}{Intuitively, Theorem~1 follows from Theorem~2, so we prove in this section Theorem~2.}
	We prove that it is impossible for an algorithm $\mathcal{A}$ that can obtain $\mathcal{L}$ (the list of non-suspected processes) to implement RTBRB-Validity and RTBRB-Timeliness, when $\mathcal{A}$ cannot change which processes are in $\mathcal{L}$. 
	
	By contradiction assume that an algorithm $\mathcal{A}$ implements RTBRB-Validity and RTBRB-Timeliness.  The $\Delta$ of the RTBRB-Timeliness is a fixed duration; hence any process can send a finite number of messages within that $\Delta$, e.g., $M$ messages. We now compute the probability that a correct process $p_i \in \mathcal{L}$ loses all $M$ messages sent to another process in $\mathcal{L}$. Recall that $P_{ij}(t)$ is the probability that link $l_{ij}$ loses a message at time $t$ $\forall i\neq j$. Let $P_{ij}(t\cap t')$ be the probability that link $l_{ij}$ loses the messages (if any is sent) at time $t$ and time $t'$. Since $0<P_{ij}(t)\leq 1$ $\forall t$, then 
	\begin{equation}\label{eq:p(x|t)}
	0<P_{ij}(t)=\frac{P_{ij}(t\cap t')}{P_{ij}(t'|t)}\leq 1~\forall~t',~t.
	\end{equation}
	By~\eqref{eq:p(x|t)}, $P_{ij}(t'|t)>0$ (and $0<P_{ij}(t\cap t')\leq 1$). By induction, we have $P_{ij}(t'|t,t_1,...,t_x)>0$ $\forall~t'>~t,t_{x}$. 
	Denote by $\mathsf{B_{ij}(t)}$ the event that link $l_{ij}$ loses all messages (if any is sent) during the interval $t+\Delta$. Let $t_{x}$ denote the times at which $p_i$ sends a message in $[t+\Delta]$. The probability of $\mathsf{B_{ij}(t)}$ happening is:
	\begin{equation*}
	\resizebox{\hsize}{!}{$
		\begin{split}
		Pr(\mathsf{B_{ij}(t)})&=P_{ij}(t_1 \cap t_2 \cap ...\cap t_M)\\&=P_{ij}(t_1)\times P_{ij}(t_2|t_1)\times ...\times P_{ij}(t_M|t_1,...,t_{M -1})>0.
		\end{split}$}
	\end{equation*}
	Given $0<P_{ij}(t)\leq 1~\forall t$ and $P_{ij}(t'|t,t_1,...,t_x)>0$ $\forall~t'>~t,t_{x}$, then we have $0<Pr(\mathsf{B_{ij}(t)})\leq 1$; there is a positive probability that the link connecting $p_i$ to $p_j$ ($p_i,p_j\in\mathcal{L}$) loses all messages sent by $p_i$ in any finite duration. Assuming independence between links, then there is a positive probability that all messages sent by $p_i$ in any duration $\Delta$ and to any finite number of processes are lost. Assume that $p_i$ is the process that invokes a broadcast. By RTBRB-Validity $p_i$ delivers its own message. However, no other process in $\mathcal{L}$ can deliver with probability 1, $p_i$'s message within any fixed duration after being broadcast.
	
	\section{\textit{RT-ByzCast} Proof of Correctness}\label{proof of correctness}
	We prove in what follows that our \textit{RT-ByzCast} algorithm guarantees all RTBRB properties (defined in Section~III).
	\subsection{RTBRB-Validity}
	
	\textcolor{red}{Consider that a correct process $p_i$ broadcasts message $m$ at round $r$. 
		\begin{lemma} \label{lemma: crashitself}
			At round $r+\mathbb{R}$, either: (i) $p_i$ can verify that at least $2f+1$ processes received its broadcast message, or (ii) $p_i$ crashes itself.
		\end{lemma}
		\begin{proof}
			$p_i$ is a non-Byzantine process and hence it does not send different messages (relative to some broadcast instance) to different processes. By Case 1 of Definition~7, $p_i$  crashes itself in any round $r_{curr}\geq r +\mathbb{R}$, if $p_i$ has received $\leq 2f$ distinct signatures (including its own) on messages echoed from its broadcast. Therefore, at round $r +\mathbb{R}$, $p_i$ either crashes itself (hence $p_i$ is no longer correct) or has at least $2f+1$ process signatures. 
		\end{proof}}
		
		\textcolor{red}{Validity is a property that concerns a correct sender, hence the case when $p_i$ does not crash itself. Following from Lemma~\ref{lemma: crashitself}, if $p_i$ does not crash itself by round $r+\mathbb{R}$, then $p_i$ can verify that at least $2f+1$ processes received its broadcast message, in which case (lines 25-27 of Algorithm 2) $p_i$ executes Function 4. Accordingly $p_i$ delivers $m$ and begins to send deliver messages as of round $r+\mathbb{R}+1$ (lines 1-4 of Function~4).} 
		
		\textcolor{red}{Following from Case 2 of Definition~7, $p_i$ remains correct (does not crash itself) at round $r+2\mathbb{R}+1$ if at least $2f+1$ processes sign its deliver message. In other words, message $m$ which is broadcast by a correct process $p_i$ (that does not crash itself within $[r,r+3\mathbb{R}]$) is eventually delivered by some correct process (in this case $p_i$ itself), proving validity.}

		\subsection{RTBRB-No Duplication}
		A process $p_j$ in our \textit{RT-ByzCast} algorithm delivers a message $m$ broadcast by $p_i$, if and if only if  $p_j$ has not previously delivered any $m$ relative to $p_i$. This is ensured by line 1 of Function~4.
		\subsection{RTBRB-Integrity}
		Assume that $p_j$ is a correct process that delivers a message $m$ relative to a correct process $p_i$. By lines 25-27 of Algorithm~2, this is possible only if $p_j$ received on some \texttt{Echo($(pi,m)...$)} message signatures relative to at least $2f+1$ processes.

		This means at least $f+1$ correct processes are transmitting some \texttt{Echo($(pi,m)...$)} messages
		In our algorithm, a process $p_k$ transmits an echo message with sender $p_i$ and message $m$ in two cases: (a)~if $p_k$ receives a \texttt{RTBRB-broadcast($(pi,m),...$)} message from $p_i$ itself (lines 8-11 of algorithm~2), or (b)~if $p_k$ receives some \texttt{Echo($(pi,m)...$)} message signed by~$p_i$ (lines 13-24 of Algorithm~2).
		
		In both cases (a) and (b), $p_k$ is sure that $p_i$ did indeed broadcast $m$.
		\subsection{RTBRB-Agreement}
		\textcolor{red}{\begin{lemma}[No two correct processes deliver different messages]\label{lemma:No two correct processes deliver different messages }
				If a correct process $p$ delivers message $m$ and a correct process $q$ delivers message $m'$, then $m=m'$.
			\end{lemma}
			\begin{proof}
				By lines (32-34 and 45-48 of Algorithm~2), any process that executes the  \textit{deliver-message()} function to deliver a message $m$ of $p_i$ should have $m$ signed by at least $2f+1$ processes. 
				In our algorithm a correct process appends its signature to $m$ (signs $m$) if it is not echoing any message other than $m$. In other words, a correct process never appends its signature to two messages $m$ and $m'$ such that $m\neq m'$. 
				Any set of $2f+1$ processes has at least one correct process in common. This means that no two correct processes deliver different messages $m$ and $m'$ (where $m\neq m'$) of $p_i$. 
			\end{proof}
			\begin{lemma}\label{lemma: all correct deliver}
				If some correct process $p_j$ delivers message $m$ then all correct processes eventually deliver some message $m'$.
			\end{lemma}
			\begin{proof}
				Let $p_j$ be the only process that delivers message $m$, say at round $r$. From Function 4, $p_j$ starts to send only \texttt{Deliver$_{p_j}$($(m)$)} messages from round $r+1$ till round $r+1+2\mathbb{R}$. At the end of round $r+1+\mathbb{R}$ process $p_j$ crashes itself if it does not receive \texttt{Deliver$_{...}$($(m)$)} signed by at least $2f+1$ processes (Case 2 of Definition~7).  Since $p_j$ is considered to be correct (hence does not crash itself), then this means that at least $2f+1$ processes have received \texttt{Deliver$_{...}$($(m)$)} by round $r+1+\mathbb{R}$. By lines 51-60 of Algorithm~2 any correct process that receives \texttt{Deliver$_{...}$($(m)$)} delivers (has delivered) some message. this means that from round $r+1+\mathbb{R}$ till round $r+1+2\mathbb{R}$ (included) at least $f+1$ correct processes send only \texttt{Deliver$_{p_j}$($(...)$)} messages. Thus, all other correct processes can gather at most $2f$ signatures on any messages sent during $[r+1+\mathbb{R},r+1+2\mathbb{R}]$ (in which case they crash themselves by Case 4 of Definition~7), unless they receive some \texttt{Deliver$_{p_j}$($(m)$)}. If they receive \texttt{Deliver$_{p_j}$($(m)$)} they deliver message $m$. Hence all correct processes (non-Byzantine processes that do not crash themselves) deliver some message.
			\end{proof}}
			
			The proof follows then from Lemma~\ref{lemma:No two correct processes deliver different messages } and Lemma~\ref{lemma: all correct deliver}.
			\subsection{RTBRB-Timeliness}
			
			\textcolor{red}{Consider that a non-Byzantine process $p_i$ broadcasts message $m$ at round $r$. and Assume that $p_i$ is correct, that is $p_i$ does not crash itself in $[r,r+\Delta]$.
				\begin{lemma}\label{lemma:f+1 send echo}
					If a correct process $p_i$ broadcasts a message $m$ at round $r$, then at least $f+1$ correct processes receive $m$ by round $r+\mathbb{R}$.
				\end{lemma}
				\begin{proof}
					Following from Lemma~\ref{lemma: crashitself}, at round $r+\mathbb{R}$, $p_i$ can verify that at least $2f+1$ processes have received its broadcast, otherwise $p_i$ crashes itself. In other words, if $p_i$ is still alive at round $r+\mathbb{R}$, then this means that at least $2f+1$ processes received $p_i$'s broadcast message out of which at least $f+1$ are correct. 
				\end{proof}
				Since process $p_i$ at round $r+\mathbb{R}$ collected $2f+1$ signatures, it delivers $m$ as indicated by Function~4 (lines 30-35 of Algorithm~2).
				\begin{lemma}\label{lemma: f+1 deliver}
					If a correct process $p_i$ broadcasts a message $m$ at round $r$, then at least $f+1$ correct processes deliver $m$ by round $r+2\mathbb{R}$.
				\end{lemma}
				\begin{proof}
					By lines (2-4) of Function~4 $p_i$ sends only \texttt{Deliver(m)} messages in all rounds $\in [r+1+\mathbb{R}, r+1+3\mathbb{R}]$. For $p_i$ not to crash itself (Case 2 of Definition~7) $p_i$ needs to collect $2f+1$ signatures on its \texttt{Deliver(m)} message by round $r+1+2\mathbb{R}$. This means if $p_i$ is does not crash itself by end of round $r+1+2\mathbb{R}$, then at least $f+1$ correct processes have received and signed $p_i's$ \texttt{Deliver(m)} message by round $r+2\mathbb{R}$. By lines (51-59) of Algorithm~2 all such correct processes deliver $m$ at lastest by round $r+2\mathbb{R}$.
				\end{proof}}
				
				\textcolor{red}{Following from Lemma~\ref{lemma: f+1 deliver}, there exists at most $f$ other correct processes at round $r+2\mathbb{R}$ that may have not delivered $m$.}
				
				\textcolor{red}{\begin{lemma}\label{lemma: all deliver}
						If a correct process $p_i$ broadcasts a message $m$ at round $r$, then all correct processes that have not delivered $m$ by round $r+2\mathbb{R}$ deliver $m$ by round $r+3\mathbb{R}$.
					\end{lemma}
					\begin{proof}
						Any process that delivers $m$ at some round $r'$, sends only \texttt{Deliver(m)} messages in all rounds $\in [r',r'+2\mathbb{R}]$ (lines 51-59 of Algorithm~2 and lines 2-4 of Function~4). Hence following from Lemma~\ref{lemma: crashitself} and Lemma~\ref{lemma: f+1 deliver}, in all rounds $[r+1+2\mathbb{R},r+1+3\mathbb{R}]$, at least $f+1$ correct processes only send \texttt{Deliver(m)} messages. Hence the rest of the correct processes (at most $f$) that have not delivered $m$ yet would (i) crash themselves at $r+1+3\mathbb{R}$ by Case 3 of Definition~7, or (ii) deliver $m$ since they heard a message from some of the $f+1$ correct processes that already delivered $m$.
					\end{proof}}

					
					As such, correct processes (that do not crash themselves) deliver message $m$ from $p_i$ by round $r+3\mathbb{R}$. Hence, \textit{RTBRB-Timeliness} is satisfied with~$\Delta=~3\mathbb{R}$.

					\section{Correlated Link Losses}\label{loss correlation}
					We now study the effect of correlated losses/omissions. To this end, we assume that links now may exhibit time-correlated losses/omissions, which might result in bursts.

					We specifically simulate links following the Gilbert-Elliot (GE) model~\cite{GE,Elliot}. The GE model, consisting of two states (see Fig.~\ref{Figure: model}), is a simple non-trivial finite state Markov chain (FSMC)~\cite{FSM}, established to capture well message loss behavior~\cite{fadingchannels,fsmwireless1,fsmwireless2}. In fact, the GE model has been empirically verified, by a large body of work~\cite{fadingchannels,transmissionschemes,BettingonGEchannels,Elliot,GE1,GE2}, as a good approximation of message losses in real-life communication scenarios. The GE model, for instance, has been used to model losses in wireless media IEEE 802.11~\cite{GE2}, wired power line networks~\cite{ABB} and other hybrid networks~\cite{packetloss,packetloss1}. The two states of the GE model (Fig.~\ref{Figure: model}), noted by \textit{good} and \textit{bad}, can for example abstract the following: the communication link between a pair of processes occupies the bad state when the packet success-rate drops below a certain ``unacceptable'' threshold or when communication delays become slower than expected (e.g., $>d$); the link occupies the good state otherwise. At any round, the link can be in exactly one of the two states: the good state or the bad state. A link in the good state delivers messages (if any is sent) in a reliably and timely fashion; however, if the link is in the bad state then messages are either lost or delayed ($>d$). In other words messages sent on a link, which is in the bad state, are omitted. Every link \textbf{\textit{transitions}} across rounds, i.e., at the beginning of every round a link moves to its new state, which can be the same state it existed in or the other state.
					
					\begin{figure}[H]
						\centering
						\includegraphics[scale=0.31]{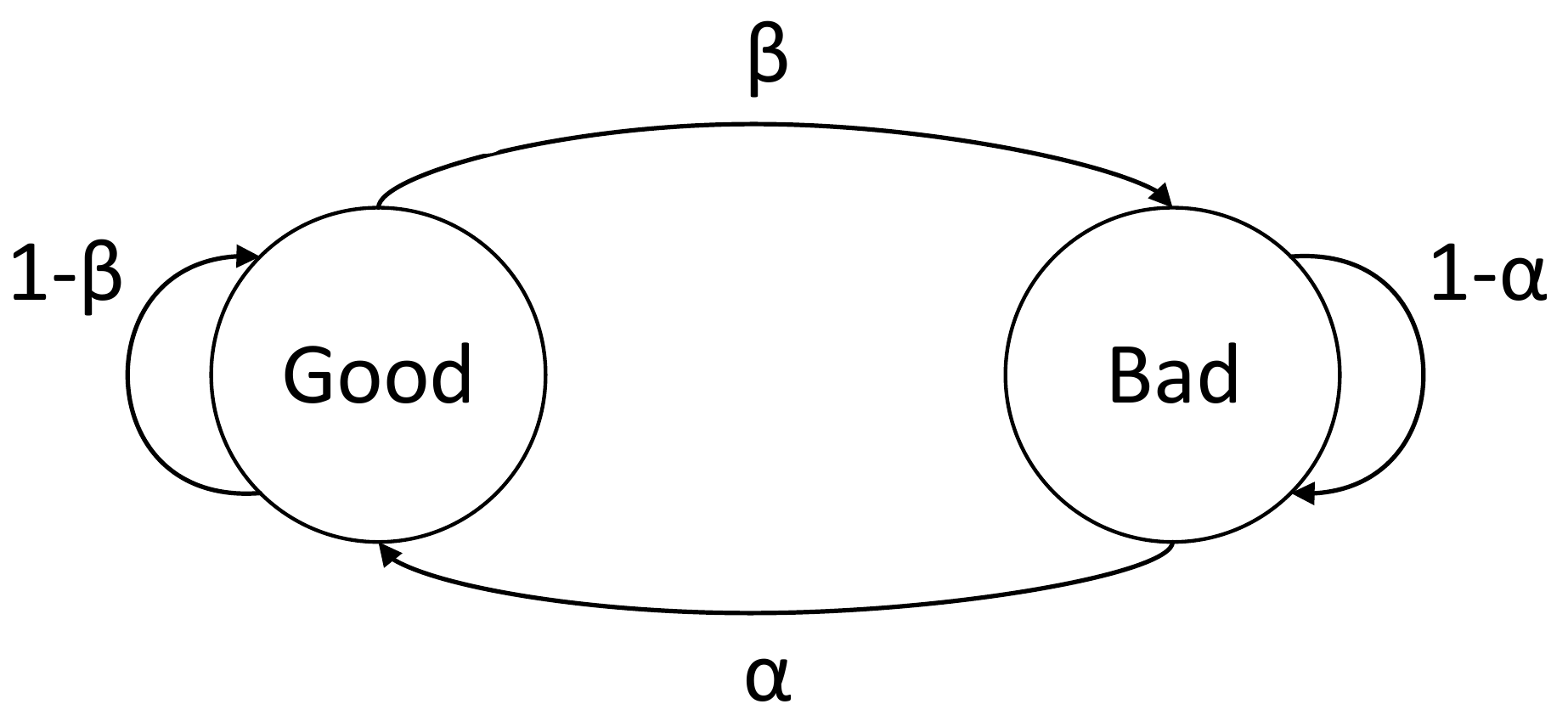}
						\vspace{-.4cm}
						\caption{A time-varying communication link under the 2-state GE model}\label{Figure: model}
						\vspace{-.2cm}
					\end{figure}
					
					For example, given the link is in the good state at some round, it will remain in the same state at the next round with probability $1-\beta$ and will move to the bad state with probability $\beta$. Similarly if the link state is bad at some round, it will remain bad at the next round with probability ($1-\alpha$) and will shift to good with probability~$\alpha$. Since links transition according to the transition probabilities (i.e., $\alpha$ and $\beta$), different links can exist in different states in a given round, even if all links start from the same state and are governed by the same $\alpha$ and $\beta$.
					
					We run our simulations again given the same $|\mathcal{C}|~\in~\{5,10,20,30,40,50,100,200\}$ processes but now considering the various transition probabilities $(\alpha,\beta) \in\{(0.8,10^{-6});(0.7,10^{-5});(0.6,10^{-3});(0.5,10^{-2});(0.4,0.1);\\(0.3,0.4);(0.2,0.6);(0.1,0.7)\}$. Namely, we assume that all links have the same transition probabilities and start from the good state. Every link then transitions to a new state at the beginning of every round. We select values of $\alpha$ and $\beta$ satisfying the positive correlation condition (a.k.a. bursty condition)~\cite{Guha}, which requires that $(1-\beta)>\alpha$. In fact, we show in Figure~\ref{busrt} the probability of having bursts of various lengths under our selected values of transition probabilities.
					
					\begin{figure}[H] 
						\begin{center}
							\includegraphics[scale=0.21]{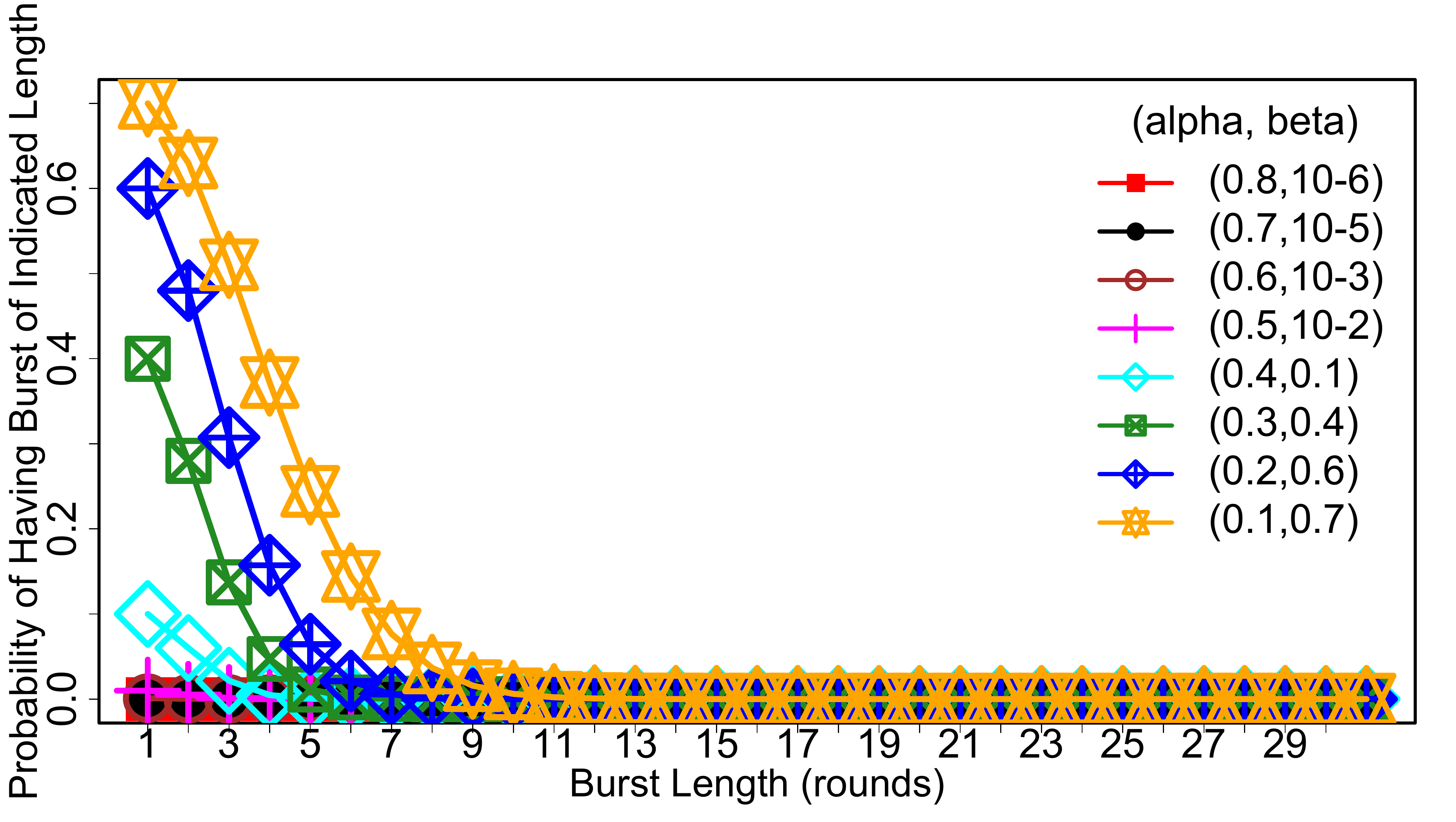}
							\begin{center}
								\vspace*{-10pt}
								\caption{The probability of bursts of variable lengths under different transition probabilities ($\alpha$ and $\beta$).}\label{busrt}
								\vspace*{-15pt}
							\end{center} 
						\end{center}
					\end{figure}
					For a given value of $|\mathcal{C}|$ and $(\alpha,\beta)$, we invoke a broadcast at one of the processes and record, after $\mathbb{R}$ rounds of communication, if any process does not receive $|\mathcal{C}|$ signatures on the value being broadcast. We repeat such an instance $10^6$~times.
					
					\begin{figure}[htbp] 
						\begin{center}
							\includegraphics[scale=0.2]{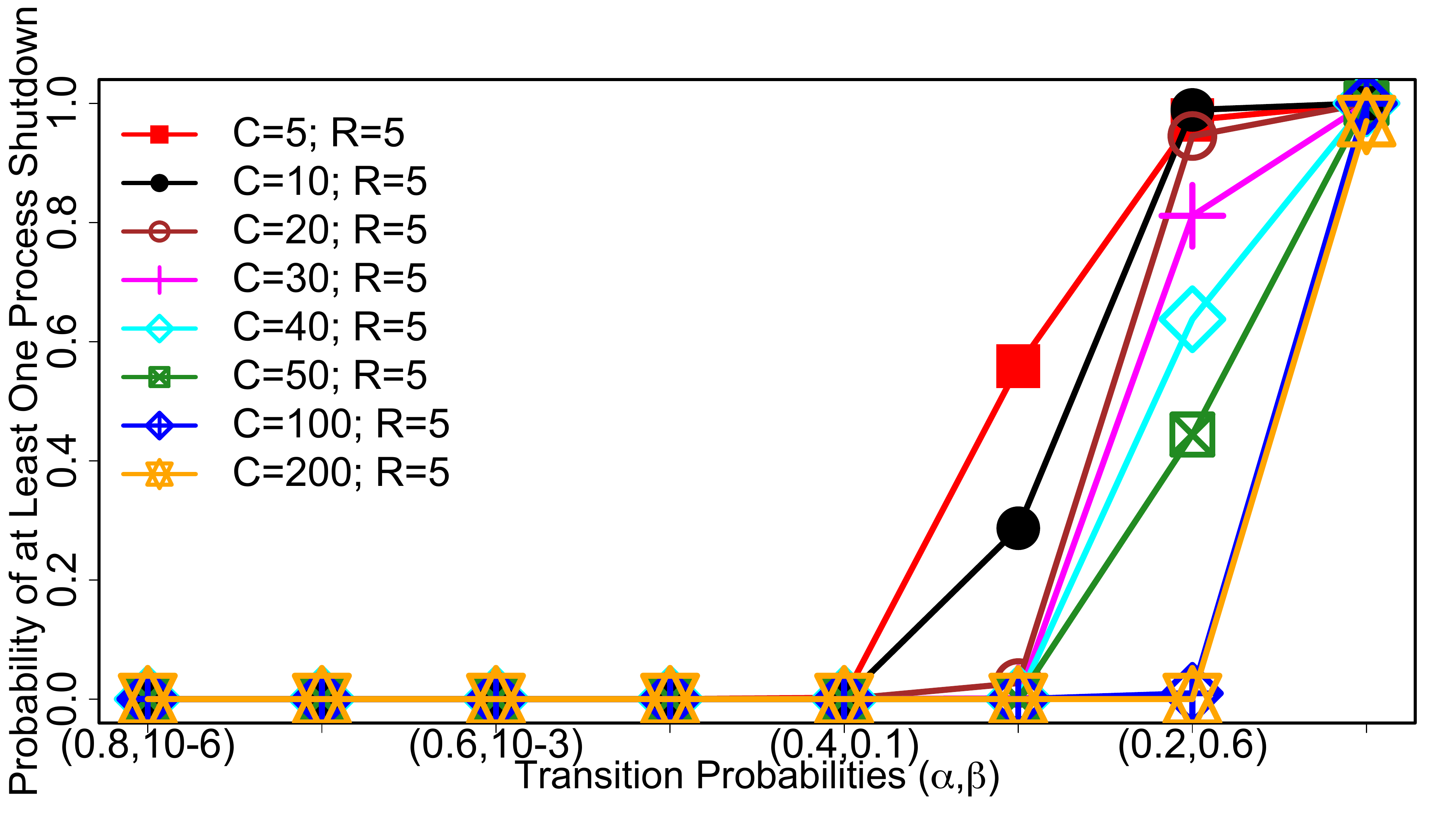}
							\begin{center}
								\vspace*{-10pt}
								\caption{The probability of a correct process crashing itself in a system after $10$ communication rounds ($\mathbb{R}=5$) under correlated link losses.}\label{pshut5GE}
								\vspace*{-15pt}
							\end{center} 
						\end{center}
					\end{figure}
					\begin{figure}[htbp] 
						\begin{center}
							\includegraphics[scale=0.2]{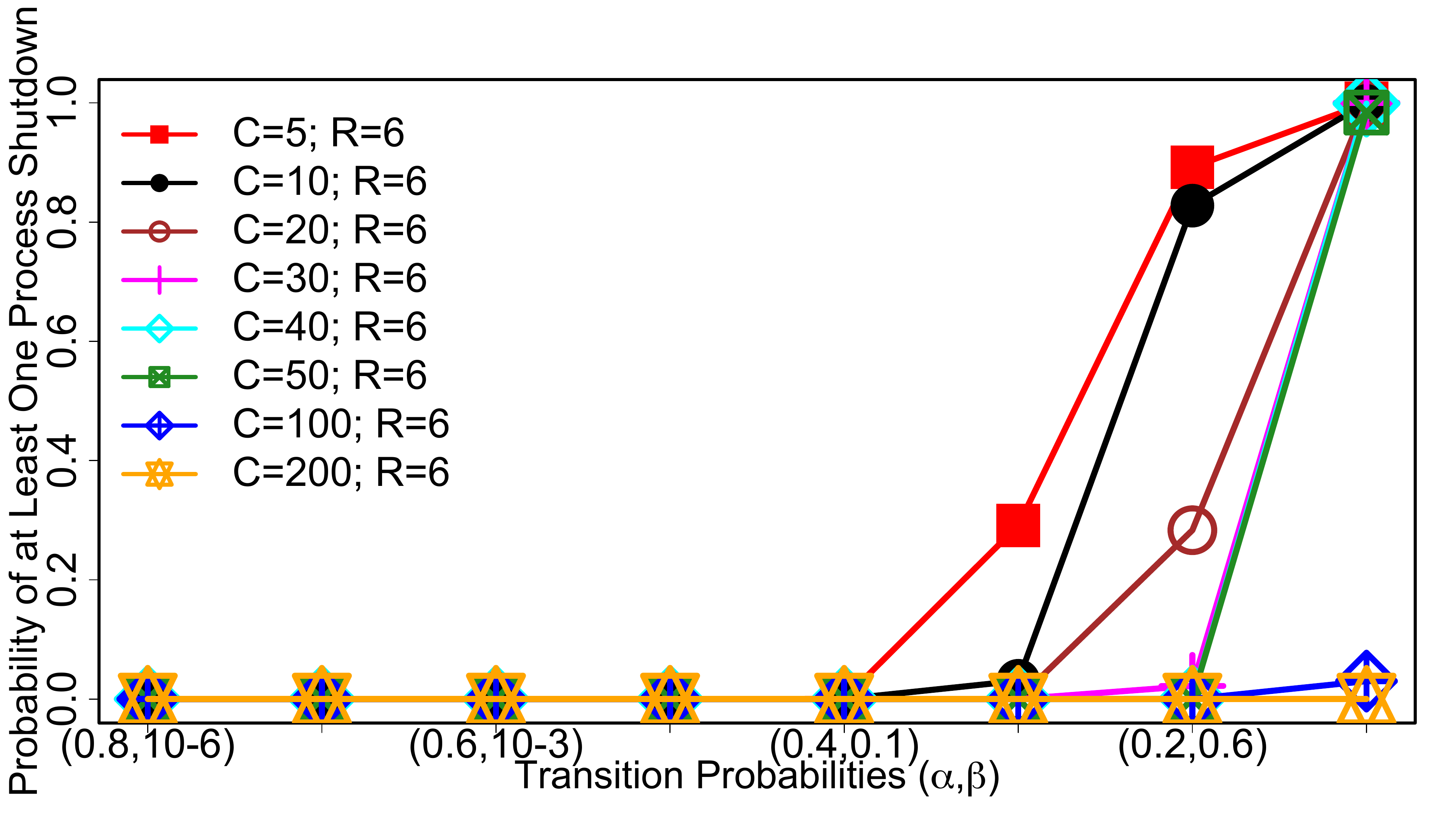}
							\begin{center}
								\vspace*{-10pt}
								\caption{The probability of a correct process crashing itself in a system after $10$ communication rounds ($\mathbb{R}=6$) under correlated link losses.}\label{pshut6GE}
								\vspace*{-15pt}
							\end{center} 
						\end{center}
					\end{figure}
					\begin{figure}[htbp] 
						\begin{center}
							\includegraphics[scale=0.2]{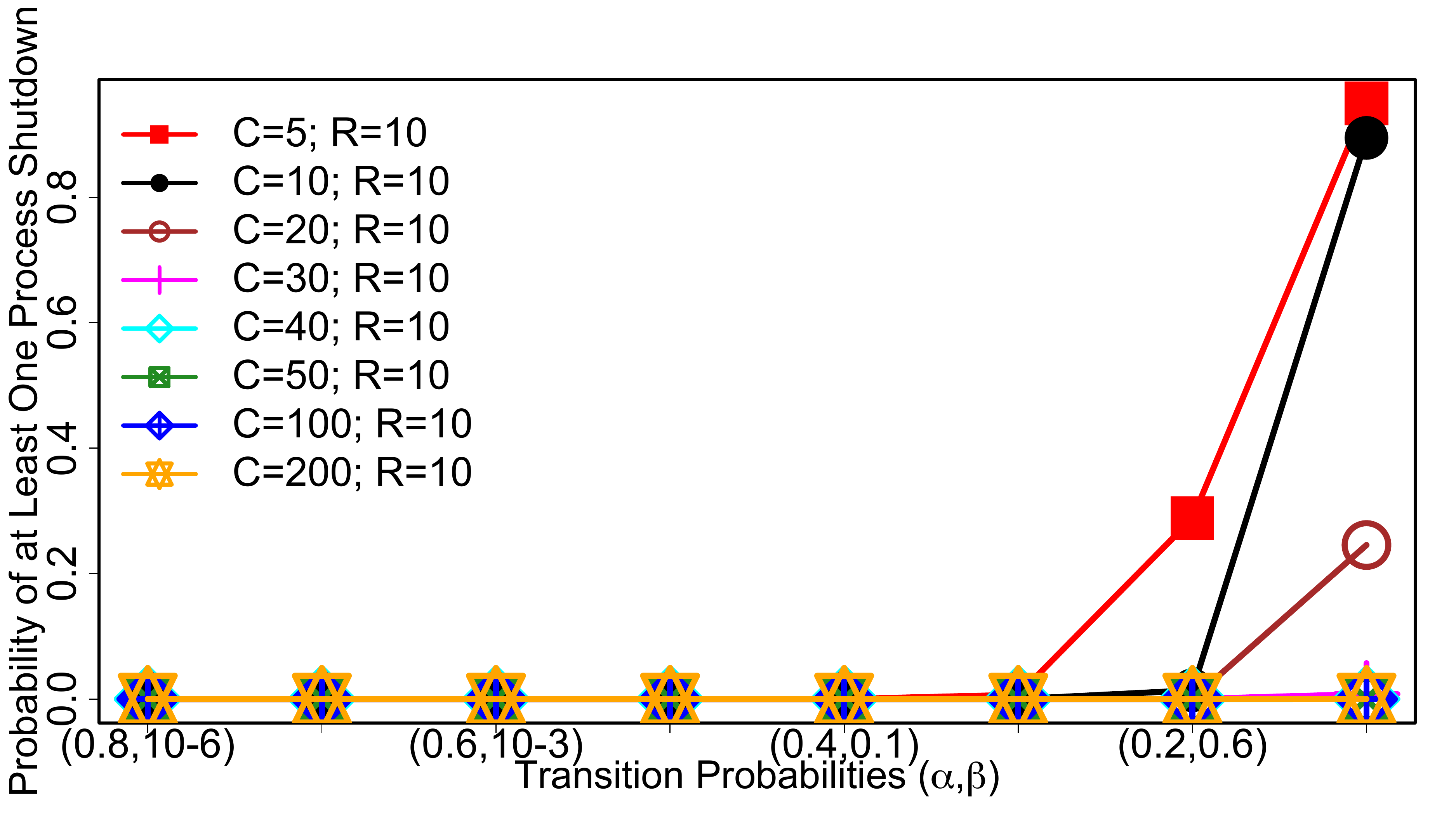}
							\begin{center}
								\vspace*{-10pt}
								\caption{The probability of a correct process crashing itself in a system after $10$ communication rounds ($\mathbb{R}=10$) under correlated link losses.}\label{pshut10GE}
								\vspace*{-15pt}
							\end{center} 
						\end{center}
					\end{figure}
					\begin{figure}[htbp] 
						\begin{center}
							\includegraphics[scale=0.2]{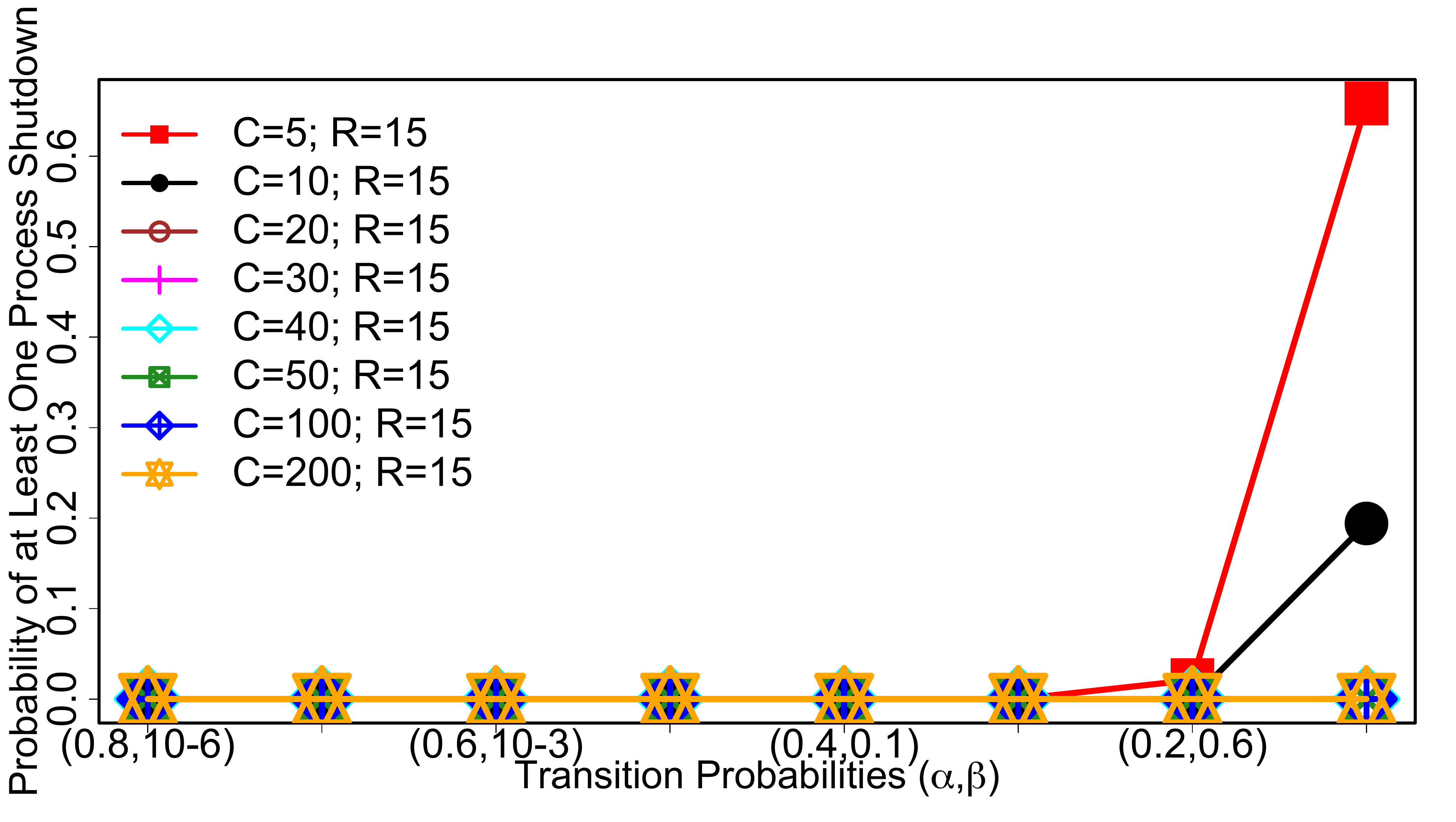}
							\begin{center}
								\vspace*{-10pt}
								\caption{The probability of a correct process crashing itself in a system after $15$ communication rounds ($\mathbb{R}=15$) under correlated link losses.}\label{pshut15GE}
								\vspace*{-15pt}
							\end{center} 
						\end{center}
					\end{figure}
					We~report our results showing: $$\resizebox{0.95\hsize}{!}{$\frac{\text{num. of instances in which some correct process crashes itself}}{10^6}$},$$ for $\mathbb{R} \in~\{5,6,10,15\}$ rounds respectively. We select these values of $\mathbb{R}$ to show that, despite loss correlation, a well chosen value of the slack window still allows our \textit{RT-ByzCast} algorithm to provide the intended reliability.
					
					Our results in Figure~\ref{pshut5GE}, Figure~\ref{pshut6GE}, Figure~\ref{pshut10GE}, and Figure~\ref{pshut15GE} show similar conclusions as the case of independent message loss. Namely, the probability of a correct process crashing itself decreases (1) as the number of non-Byzantine processes in the system increases, and (2) as the number of communication rounds (i.e., the value of $\mathbb{R}$) increases. 
					
					\section{Reliability Evaluation}~\label{sec: Reliability Eval.}
					In this section we illustrate the probability of a correct process shutting down. Unlike the graphs in Section~VI.A, we fix $|C|$ and vary the window size ($\mathbb{R}$). Our results in Figure~\ref{pshutc5}, Figure~\ref{pshutc20}, Figure~\ref{pshutc50}, and Figure~\ref{pshutc200} show that increasing the number of communication rounds (i.e., the value of $\mathbb{R}$) increases the reliability of our \textit{RT-ByzCast} algorithm for any number of correct processes. In fact, with $\mathbb{R}=10$ the probability of a process crashing itself becomes negligible even with up to $60\%$ losses/omissions rate.
					\begin{figure}[htbp] 
						\begin{center}
							\includegraphics[scale=0.2]{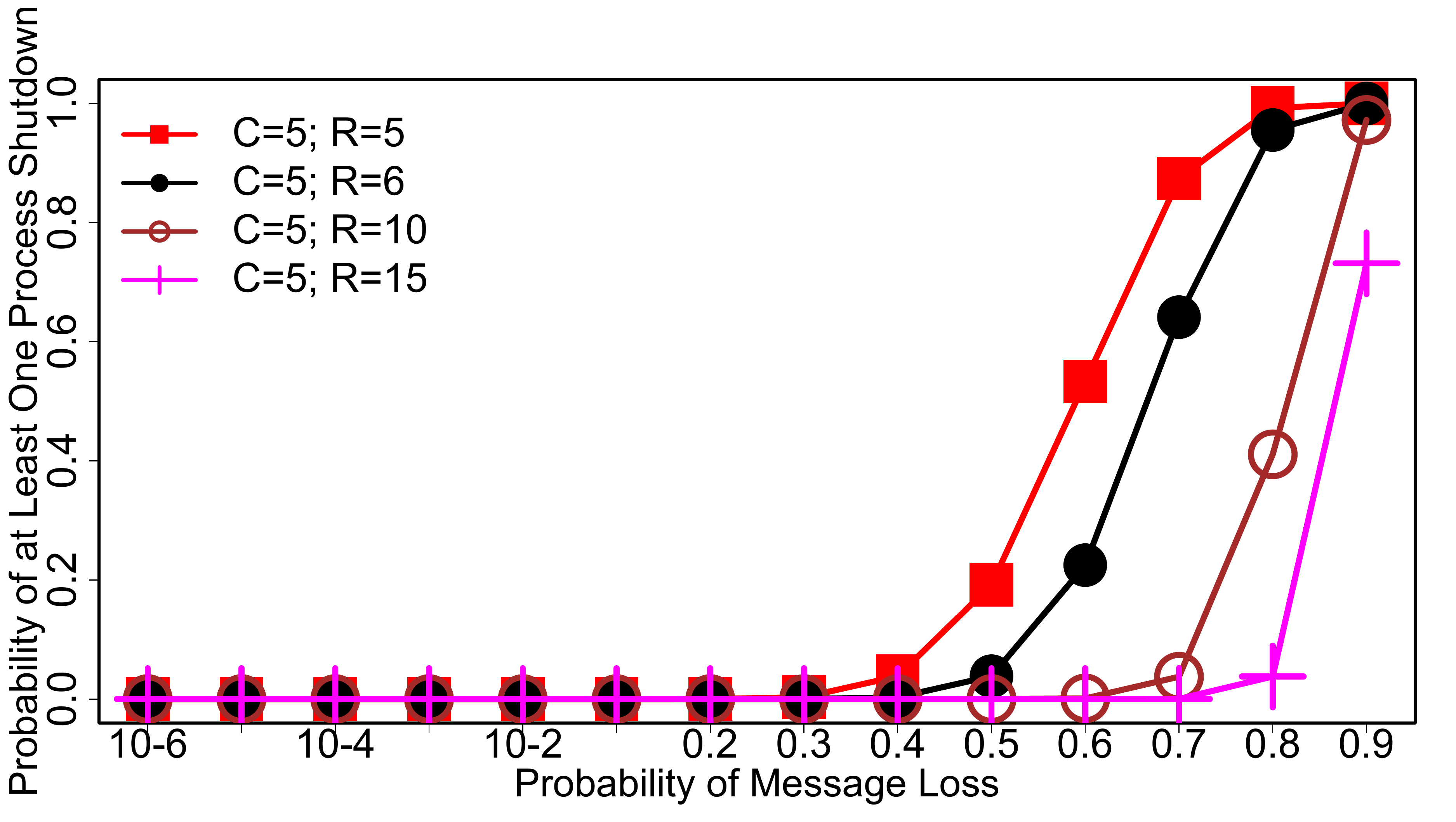}
							\begin{center}
								\vspace*{-10pt}
								\caption{The probability of a correct process crashing itself in a system of $|C|=5$ w.r.t. a varying number communication rounds.}\label{pshutc5}
								\vspace*{-15pt}
							\end{center} 
						\end{center}
					\end{figure}
					\begin{figure}[htbp] 
						\begin{center}
							\includegraphics[scale=0.2]{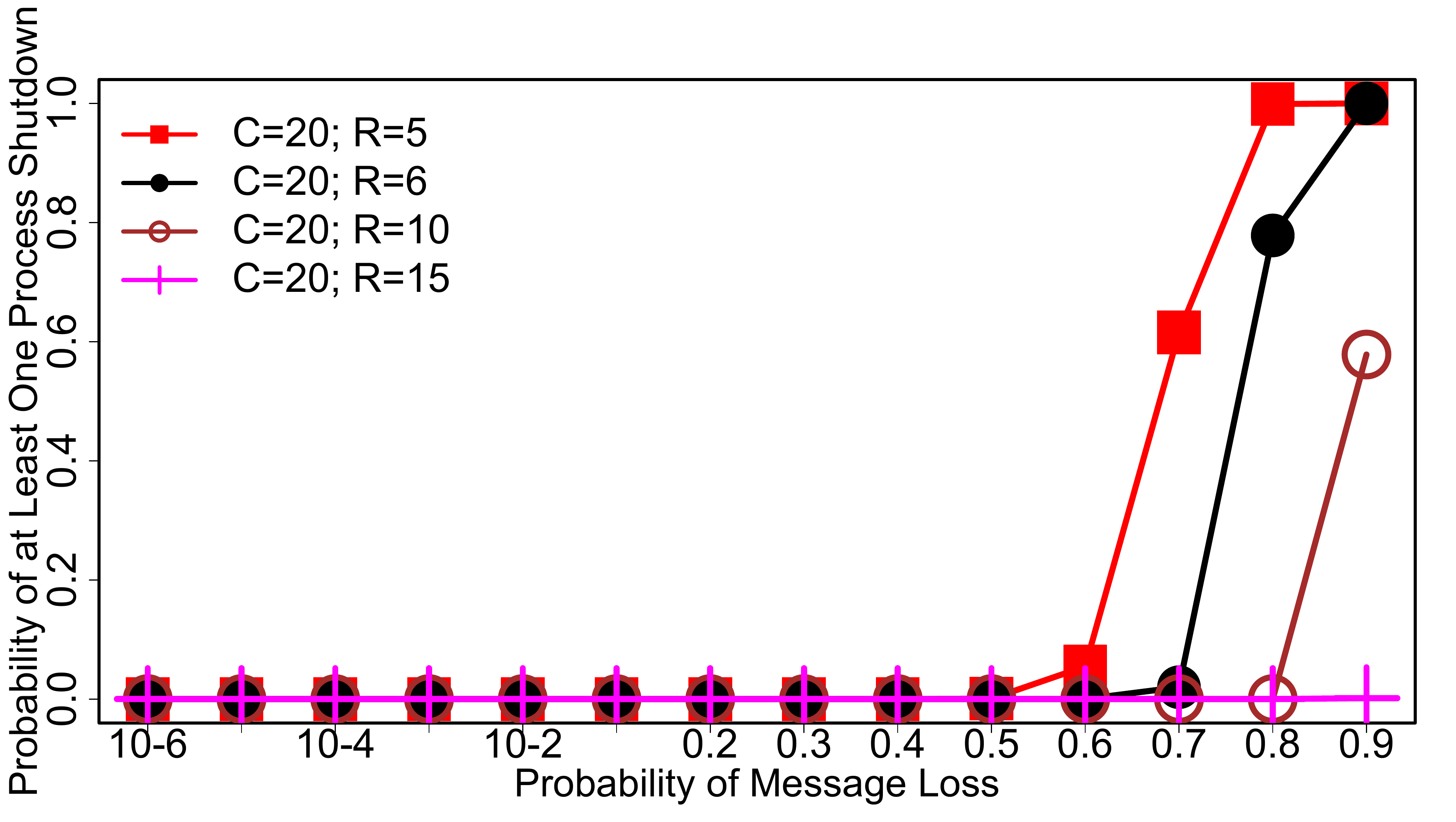}
							\begin{center}
								\vspace*{-10pt}
								\caption{The probability of a correct process crashing itself in a system of $|C|=20$ w.r.t. a varying number communication rounds.}\label{pshutc20}
								\vspace*{-15pt}
							\end{center} 
						\end{center}
					\end{figure}
					\begin{figure}[htbp] 
						\begin{center}
							\includegraphics[scale=0.2]{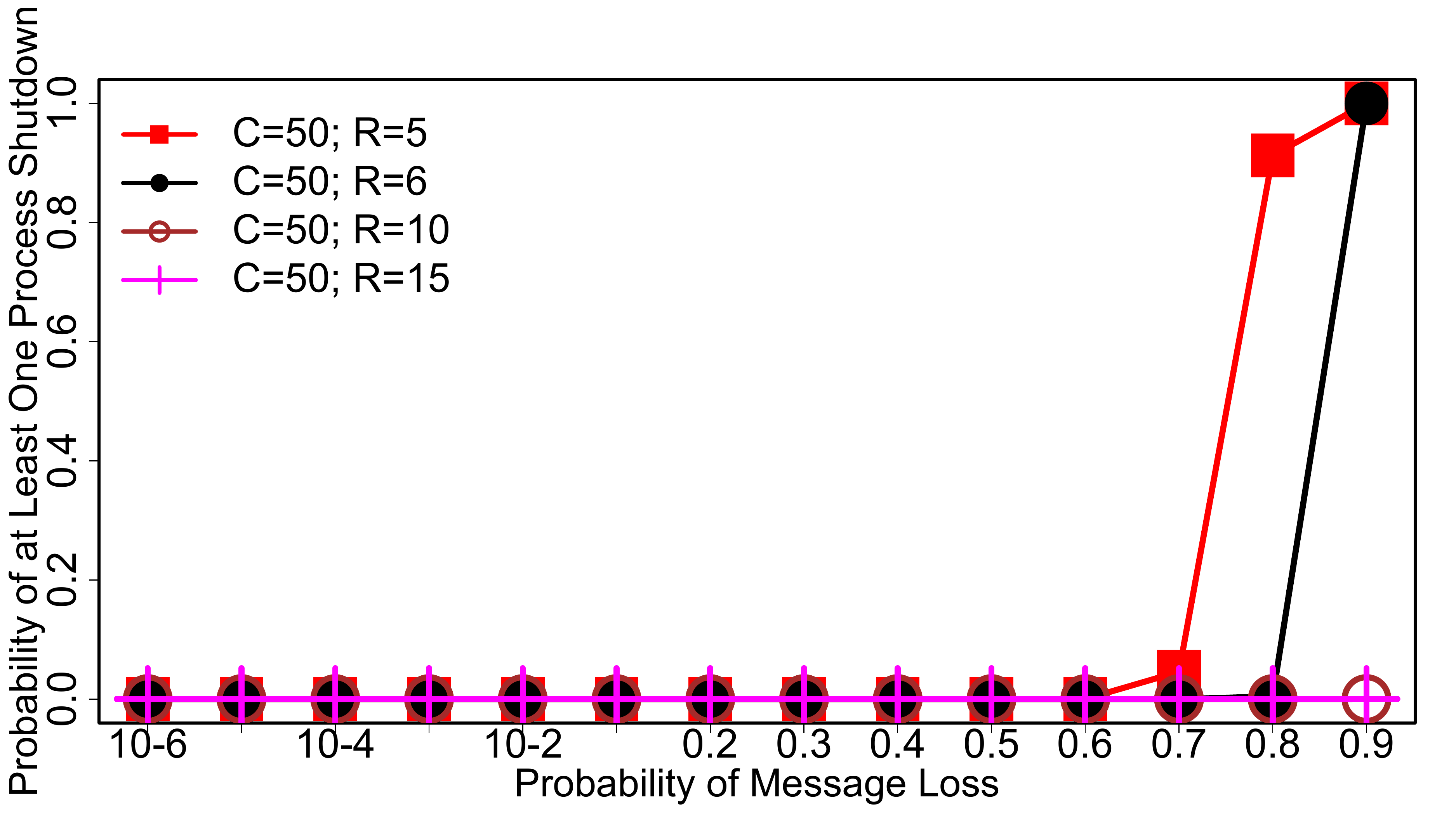}
							\begin{center}
								\vspace*{-10pt}
								\caption{The probability of a correct process crashing itself in a system of $|C|=50$ w.r.t. a varying number communication rounds..}\label{pshutc50}
								\vspace*{-15pt}
							\end{center} 
						\end{center}
					\end{figure}
					\begin{figure}[htbp] 
						\begin{center}
							\includegraphics[scale=0.2]{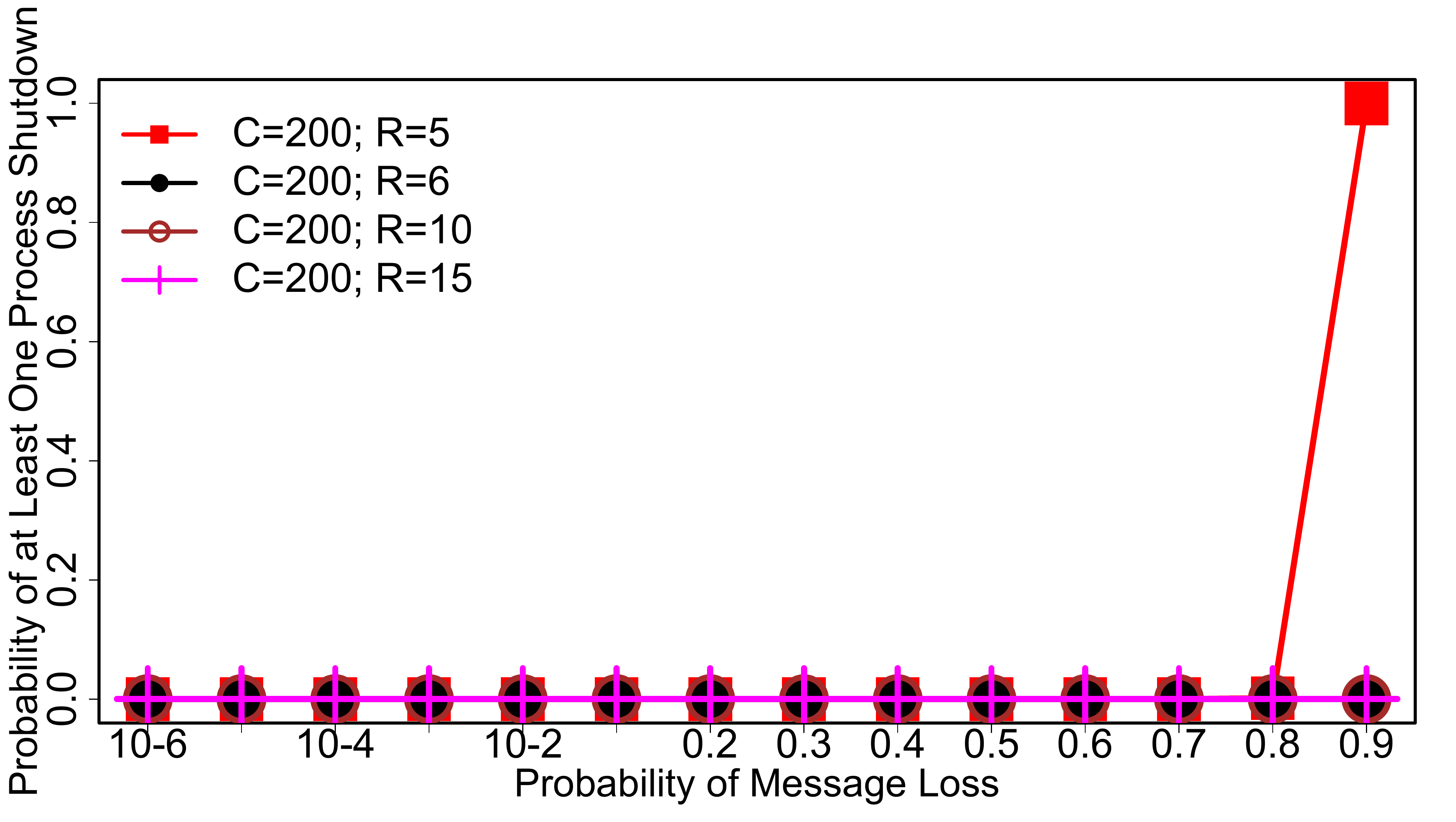}
							\begin{center}
								\vspace*{-10pt}
								\caption{The probability of a correct process crashing itself in a system of $|C|=200$ w.r.t. a varying number communication rounds..}\label{pshutc200}
								\vspace*{-15pt}
							\end{center} 
						\end{center}
					\end{figure}

					\section{\textcolor{red}{\textit{RT-ByzCast} Bandwidth Consumption}}\label{sec: RT-ByzCast Bandwidth Consumption}
					
					\textcolor{red}{Finally, we evaluated the protocol's peak bandwidth consumption, both in terms of reception and in emission. Figure~\ref{fig:bdw128} and Figure~\ref{fig:bdw1M} respectively present the peak bandwidth consumptions during the broadcast of 128 bits and 1 Mbits messages, depending on the message loss probability. We observed that correct nodes are using more bandwidth when there are Byzantine nodes in the system. We therefore present the bandwidth consumptions when the system contains the maximum number of Byzantine nodes tolerated. We use plain lines to represent the nodes' bandwidth usage in reception, and dashed lines for their bandwidth usage in emission. For a given system size, the difference between those two lines is due to message losses, and is accentuated by the presence of Byzantine nodes.}
					
					\textcolor{red}{From these figures, we could make two observations. First, as expected, the nodes' bandwidth consumption increases with the system's size and with the message's size. Second, the bandwidth consumptions of nodes tend to decrease when the message loss probabilities increase. 
						For example, Figure~\ref{fig:bdw128} shows that broadcasting a 128 bits message in a system of 100 nodes consume around 1.9 Mbits in emission, and less than 1.1 Mbits in reception. Comparatively, broadcasting a larger 1 Mbits message in a system of 100 nodes consume less than 100 Mbits in emission, and less than 75 Mbits in reception. Those numbers were obtained without any optimizations, and we leave for future works the improvement of \textit{RTByzCast}'s bandwidth consumption. Possible ideas include reducing the nodes' fanout, decreasing the bandwidth's footprint of digital signatures, and the avoidance of redundant messages.}
					
					\begin{figure}[htbp] 
						\begin{center}
							\includegraphics[width=0.9\columnwidth]{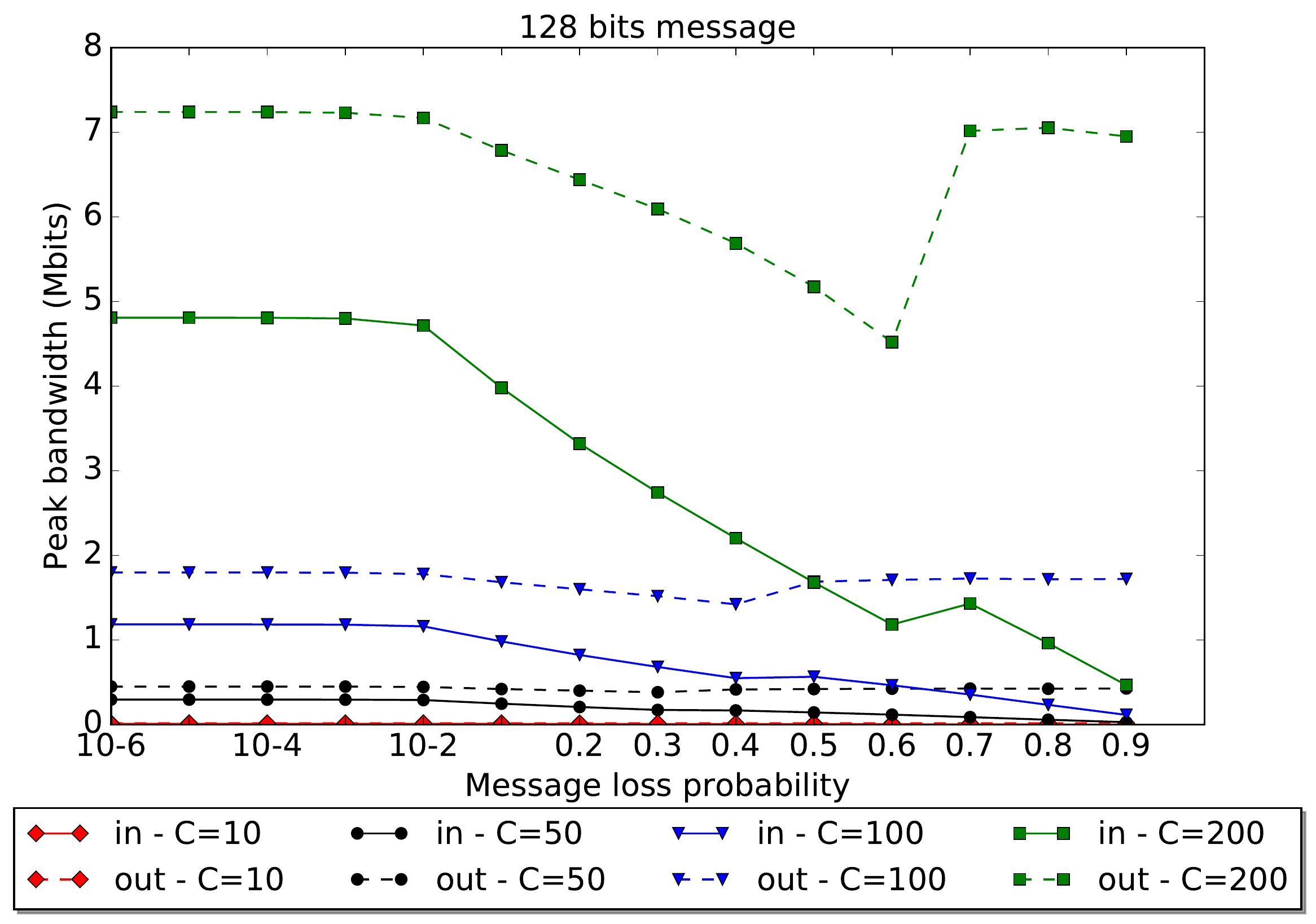}
							\begin{center}
								\vspace*{-10pt}
								\caption{\textcolor{red}{Total broadcast delay for a 128 bits-message depending on the number of nodes and the probability of message losses.}}\label{fig:bdw128}
								\vspace*{-15pt}
							\end{center} 
						\end{center}
					\end{figure}
					
					\begin{figure}[htbp] 
						\begin{center}
							\includegraphics[width=0.9\columnwidth]{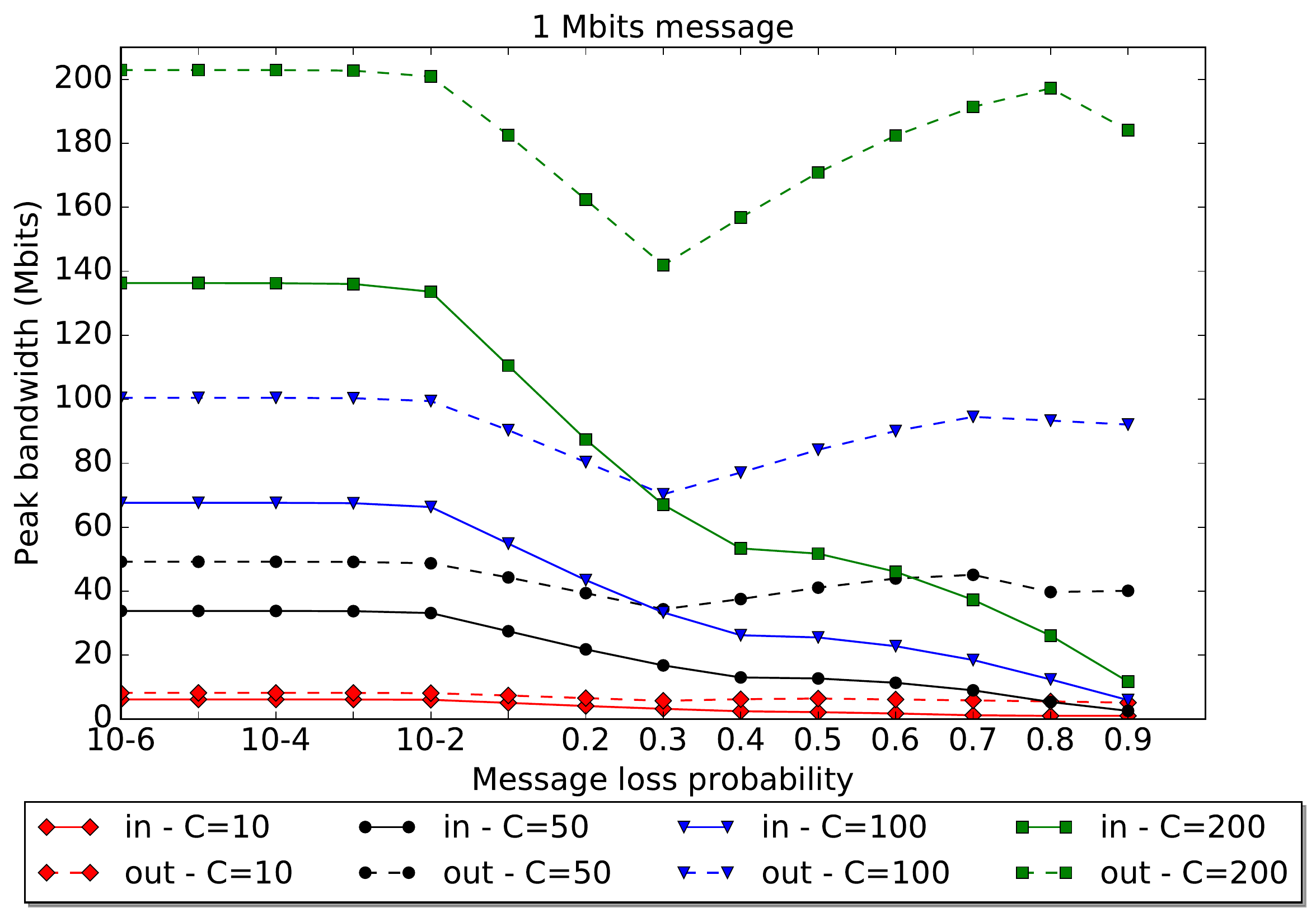}
							\begin{center}
								\vspace*{-10pt}
								\caption{\textcolor{red}{Total broadcast delay for a 1 Mbits-message depending on the number of nodes and the probability of message losses.}}\label{fig:bdw1M}
								\vspace*{-15pt}
							\end{center} 
						\end{center}
					\end{figure}

\end{document}